\documentclass[10pt,draftclsnofoot]{IEEEtran}
\onecolumn
\usepackage{amsmath, amssymb,amsthm}
\usepackage[dvips]{graphics}
\usepackage{graphicx}
\usepackage{psfrag}
\usepackage{dblfloatfix}
\usepackage{mathtools}
\DeclarePairedDelimiter{\ceil}{\lceil}{\rceil}

\allowdisplaybreaks
\newtheorem{lemma}{Lemma}
\newtheorem{theorem}{Theorem}

\theoremstyle{definition}
\newtheorem{example}{Example}
\newcommand{\RNum}[1]{\uppercase\expandafter{\romannumeral #1\relax}}
\makeatletter
\newcommand*{\rom}[1]{\expandafter\@slowromancap\romannumeral #1@}
\makeatother
\newcommand{\IX}{\underline{X}^{n}}
\newcommand{\Ix}{\underline{x}^{n}}
\newcommand{\IY}{\underline{Y}^{n}}
\newcommand{\Iy}{\underline{y}^{n}}
\newcommand{\IA}{\underline{A}^{n}}
\newcommand{\Ia}{\underline{a}^{n}}
\newcommand{\ItX}{\underline{\tilde{X}}^{n}}
\newcommand{\Itx}{\underline{\tilde{x}}^{n}}
\newcommand{\ItY}{\underline{\tilde{Y}}^{n}}
\newcommand{\Ity}{\underline{\tilde{y}}^{n}}
\newcommand{\ItA}{\underline{\tilde{A}}^{n}}
\newcommand{\Ita}{\underline{\tilde{a}}^{n}}
\newcommand{\vX}{X^{n}}
\newcommand{\vx}{x^{n}}
\newcommand{\vY}{Y^{n}}
\newcommand{\vy}{y^{n}}
\newcommand{\vA}{A^{n}}
\newcommand{\va}{a^{n}}
\newcommand{\vtX}{\tilde{X}^{n}}
\newcommand{\vtx}{\tilde{x}^{n}}
\newcommand{\vtY}{\tilde{Y}^{n}}
\newcommand{\vty}{\tilde{y}^{n}}
\newcommand{\vtA}{\tilde{A}^{n}}
\newcommand{\vta}{\tilde{a}^{n}}
\newcommand{\net}{\mathbf{BG}_t}
\newcommand{\strtyp}{\mathcal{T}_\epsilon^{(n)}}
\newcommand{\netz}{\mathbf{BG}_t=\underline{Z}_{t}}
\makeatletter
\def\blfootnote{\gdef\@thefnmark{}\@footnotetext}
\makeatother
\begin{document}
%\pagestyle{plain}
%\pagenumbering{roman}
\title{Lossless Coding of Correlated Sources with Actions}
\author{Oron Sabag, Haim H. Permuter and Asaf Cohen}

%\author{
%\IEEEauthorblockN{Oron Sabag}
%\IEEEauthorblockA{
%Ben Gurion University\\
%oronsa@ee.bgu.ac.il}
%\and
%\IEEEauthorblockN{Haim H. Permuter} \IEEEauthorblockA{
%Ben Gurion University \\
%haimp@bgu.ac.il }
%\and
%\IEEEauthorblockN{Asaf Cohen} \IEEEauthorblockA{
%Ben Gurion University\\
%coasaf@bgu.ac.il }
%}
\maketitle %\vspace{-1.4cm}
\begin{abstract}
This work studies the problem of distributed compression of correlated sources with an action-dependent joint distribution. This class of problems is, in fact, an extension of the Slepian-Wolf model, but where cost-constrained actions taken by the encoder or the decoder affect the generation of one of the sources. The purpose of this work is to study the implications of actions on the achievable rates.

In particular, two cases where transmission occurs over a rate-limited link are studied; case A for actions taken at the decoder and case B where actions are taken at the encoder. A complete single-letter characterization of the set of achievable rates is given in both cases. Furthermore, a network coding setup is investigated for the case where actions are taken at the encoder. The sources are generated at different nodes of the network and are required at a set of terminal nodes, yet transmission occurs over a general, acyclic, directed network. For this setup, generalized cut-set bounds are derived, and a full characterization of the set of achievable rates using single-letter expressions is provided. For this scenario, random linear network coding is proved to be optimal, even though this is not a classical multicast problem.
Additionally, two binary examples are investigated and demonstrate how actions taken at different nodes of the system have a significant affect on the achievable rate region in comparison to a naive time-sharing strategy.
\end{abstract}
\blfootnote{This work was supported by the Israel Science Foundation, the ERC starting grant and the European Commission in the framework of the FP7 Network of Excellence in Wireless COMmunications (NEWCOM$\#$).
This paper will be presented in part at the 2014 IEEE International Symposium on Information Theory, Honolulu, HI, USA.
O. Sabag and H. H. Permuter are with the department of Electrical and Computer Engineering, Ben-Gurion University of the Negev, Beer-Sheva, Israel (oronsa@post.bgu.ac.il, haimp@bgu.ac.il).
A. Cohen is with the department of Communication Systems Engineering, Ben-Gurion University of the Negev, Beer-Sheva, Israel (coasaf@bgu.ac.il).}
\section{Introduction}\label{section:intro}
The field of distributed encoding and joint decoding of correlated information sources is fundamental in information theory. In their seminal work, Slepian and Wolf (SW)\cite{Slepian_wolf_73_source_coding} showed that the total rate used by a system which distributively compresses correlated sources is equal to the rate that is used by a system that performs joint compression. An extension of this model for general networks was studied by Ho \emph{et al.} \cite{random}, who showed that this property is maintained using a novel coding scheme, Random Linear Network Coding (RLNC).

In past studies, the joint distribution of the sources has been perceived as given by nature; however, what if the system can take actions that affect the generation of sources?

For instance, consider a sensor network where measurements of temperature and pressure sensors are required at a set of terminal nodes. Each source symbol is acquired via a sensor and the resolution of the pressure sensors can be controlled by actions. After collecting data from the temperature sensors, we may wish to perform actions according to our needs. Based on a block of temperature measurements, actions are taken by modifying the pressure measurements' resolution. We model such a system as \textit{correlated sources with actions} with the following sources distribution: the source $X$ is a memoryless source that is distributed according to $P_{X}$, while the other source, $Y$, has a memoryless conditional distribution, $P_{Y|X,A}$, that is conditioned on the source $X$ and an action $A$.

In this paper, we cover two concepts for our model; the first is a classical multi-user setup where transmission occurs over rate-limited links. Here, actions can be performed at different nodes of the system: case A for actions that are taken at the decoder as described in Fig. \ref{fig:dec}, and case B for actions that are taken at the encoder as described in Fig. \ref{fig:enc}. In the second approach, we assume that transmission occurs over a \emph{given directed, acyclic network}. In this scenario, the case where actions are taken at the encoder is investigated. Our coding scheme combines both codes for coding of correlated sources with actions as well as Network Coding. Moreover, we define \emph{generalized cut-set bounds} for this setup, which are shown to be tight. To the best of our knowledge, actions have not been previously studied in a general network coding setup.

\begin{figure}[t]
\centering
        \psfrag{A}[][][1]{Encoder $1$}
        \psfrag{B}[][][1]{Encoder $2$}
        \psfrag{C}[][][1]{Decoder}
        \psfrag{D}[][][0.9]{$T_{1}(X^{n})\in2^{nR_X}$}
        \psfrag{E}[][][0.9]{$A^{n}(T_{1})$}
        \psfrag{F}[l][][0.9]{$T_{2}(Y^{n})\in2^{nR_Y}$}
        \psfrag{G}[][][0.9]{$\hat{X}^{n},\hat{Y}^{n}$}
        \psfrag{P}[][][1]{$P_{Y|X,A}$}
        \psfrag{X}[][][0.9]{$X^{n}$}
        \psfrag{Y}[][][0.9]{$Y^{n}$}
        \centerline{\includegraphics[height = 4.5cm]{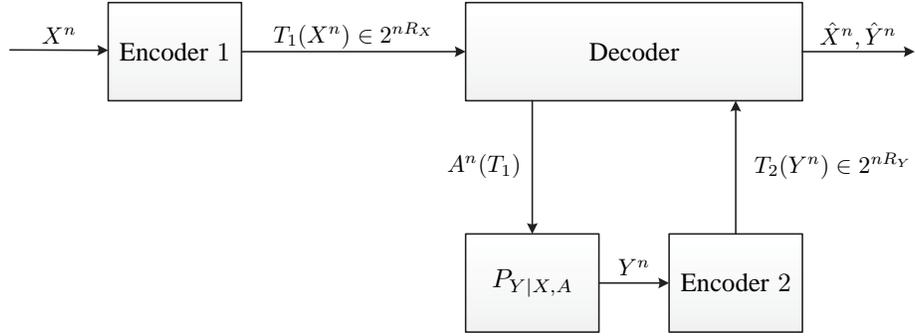}}
\caption{Case A - Correlated sources with actions taken at the decoder. The actions are based on the index $T_{1}$ sent by encoder $1$ and affect the generation of the source $Y^{n}$.}
\label{fig:dec}
\end{figure}

Specifically, the first case we consider is depicted in Fig. \ref{fig:dec}, where actions are taken at the decoder: based on its source observation $X^{n}$, which is independent and identically distributed (i.i.d.) according to $\sim P_{X}$, encoder $1$ gives an index $T_{1}(X^{n})$ to the decoder. Having received the index $T_{1}$, the decoder chooses the action sequence $A^{n}$. Nature then generates the other source sequence, $Y^{n}$, which is the output of a discrete memoryless channel $P_{Y|X,A}$, whose input is the pair $(X^{n},A^{n})$. Based on its observation $Y^{n}$,  an index $T_{2}(Y^{n})$ is sent to the decoder by encoder $2$. The reconstruction sequences $(\hat{X}^{n},\hat{Y}^{n})$ are then generated at the decoder, based on the indices that were given by the encoders. For this case, a single-letter characterization of the optimal rate region is presented in Theorem \ref{theorem:dec}.

The second case we consider is depicted in Fig. \ref{fig:enc}, where actions are taken at encoder $1$: based on its source observation $X^{n}$, which is i.i.d. $\sim P_{X}$, the first encoder chooses an action sequence $A^{n}$. The other source, $Y^{n}$, is then generated as in case A and is available at encoder $2$. Each encoder now chooses an index to be given to the decoder, based on its source observation. The reconstruction sequences $(\hat{X}^{n},\hat{Y}^{n})$ are then generated at the decoder based on the indices that were given by the encoders. This case is found to have better performance than case A, which is intuitive since in case A actions are constrained to be a function of $T_1$, while in case B actions are a function of the explicit source $X^{n}$. Moreover, in case A encoder $1$ is required to describe completely the actions' information within the index $T_1$, while in case B partial actions' information can be sent within $T_2$. In Theorem \ref{theorem:enc}, we characterize the optimal rate region for this case using single-letter terms. In Section \ref{section:examples}, we demonstrate and prove in two binary examples how performing actions at the encoder or the decoder have a significant advantage compared to a naive time-sharing strategy.

\begin{figure}[t]
\centering
        \psfrag{A}[][][1]{Encoder $1$}
        \psfrag{B}[][][1]{Encoder $2$}
        \psfrag{C}[][][1]{Decoder}
        \psfrag{D}[][][0.9]{$T_{1}(X^{n})\in2^{nR_X}$}
        \psfrag{E}[r][][0.9]{$A^{n}(X^{n})$}
        \psfrag{F}[l][][0.9]{$T_{2}(Y^{n})\in2^{nR_Y}$}
        \psfrag{G}[][][0.9]{$\hat{X}^{n},\hat{Y}^{n}$}
        \psfrag{P}[][][1]{$P_{Y|X,A}$}
        \psfrag{X}[][][0.9]{$X^{n}$}
        \psfrag{Y}[cb][][0.9]{$Y^{n}$}
        \centerline{\includegraphics[height=4.5cm]{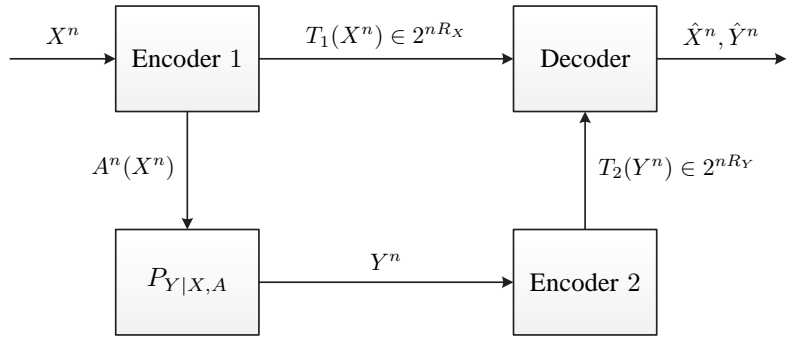}}
\caption{Case B - Correlated sources with actions taken at the encoder. The actions are based on the source $X^{n}$ and affect the generation of the source $Y^{n}$.}
\label{fig:enc}
\end{figure}
In the general network scenario, the case where actions are taken at the encoder is investigated. The setup is depicted in Fig. \ref{fig:network}. The nodes $s_1$ and $s_2$ play the role of the encoders as in case B and source generation remains the same. However, transmission occurs over a general, acyclic, directed network. Each link in the network has a known capacity, which represents a noiseless link in units of bits per unit time. Nodes in the network are allowed to perform encoding based on the messages on their input links, except for a set of terminal nodes $\tau$. Each terminal node, $t\in\tau$, is required to reconstruct both sources in a lossless manner. To characterize the set of achievable rates, we derived the conditions for which reliable communication can occur in terms of network capabilities and, lastly, proved its optimality by deriving the generalized cut-set bounds for this problem.

\begin{figure*}[b]
\centering
\begin{psfrags}
    \psfragscanon
    \psfrag{A}[][][.8]{$s_{1}$}
    \psfrag{C}[][][.8]{$s_{2}$}
    \psfrag{B}[c][][0.8]{$P_{Y|X,A}$}
    \psfrag{G}[][][0.7]{$Y^{n}$}
    \psfrag{H}[][][.7]{$\hat{X}^{n},\hat{Y}^{n}$}
    \psfrag{E}[b][][0.7]{$X^{n}$}
    \psfrag{F}[r][][0.7]{$A^{n}(X^{n})$}
\includegraphics[height=4.3cm]{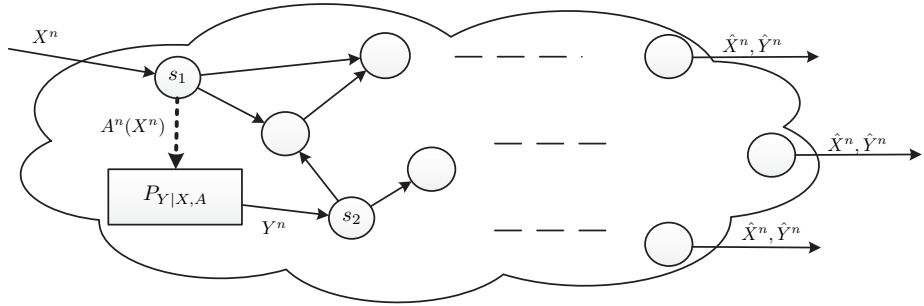}
\caption{Correlated sources in general networks with actions. Based on source $X^{n}$, node $s_{1}$ performs actions that affect the generation of $Y^{n}$. Transmission of the encoded sources occurs over an arbitrary acyclic directed network. Both sources are required at a set of terminal nodes. Note, the dashed arrow is the actions' cost-constrained link.} \label{fig:network}
\psfragscanoff
\end{psfrags}
\end{figure*}

In \cite{AhlswedeYeung_NetworkCoding_first}, it was proven by Li \emph{et al.} that linear network coding achieves optimality in multicast problems. Following this result, the RLNC approach was introduced by Ho \emph{et al.} in \cite{random} for a model of correlated sources' compression over an arbitrary network. Our model \emph{does not fall into the class of multicast problems} since no requirement for actions reconstruction is defined, yet it is very clear that the actions taken affect the rate region. Moreover, our set of achievable rates include terms of mutual information, which are not typical in multicast problems. Nevertheless, we prove that RLNC achieves optimality also in our network model. Furthermore, derivation of the achievable region for our model required an upper bound on the probability that two different inputs to a randomized linear network induce the same output at a receiver node. Calculation of these bounds, based on the result in \cite[Appendix A]{random}, led us to note that their result can be extended to a broader class of network coding problems. In Lemma \ref{lemma:bounds}, we state the upper bound and provide an alternative proof followed by an example that demonstrates how this lemma can be used in network coding problems in general, and, particularly, in our model.

The concept of actions in information theory was introduced by Weissman in \cite{Weissman_action_dependtent_state10}. The model is useful in cases where the user can perform actions that control the problem setting, such as receiving state information in channel coding or receiving side information (SI) in source coding problems. In \cite{Weissman_action_dependtent_state10}, a point to point channel with an action-dependent state was studied. Based on the input message, the transmitter was allowed to perform actions that affect the generation of states in the channel and are available at the transmitter. In \cite{6503450}, Choudhuri and Mitra studied an adaptive actions setting; actions' performance was not based only on the message but also on a causal observation of the channel state. This adaptive setup was proved to have the same performance as in \cite{Weissman_action_dependtent_state10}. Later, some extensions to the multi-user setups were considered, such as multiple access channel (MAC) with cribbing and controlled encoders by Permuter and Asnani \cite{PermuterAsnani_cribbing}, and MAC with action-dependent state information at one encoder \cite{DiksteinPermutyerShamai_ISIT12_MAC_actions} by Dikstein \emph{et al.}. In all the setups described above, considering actions was proved to increase the capacity rate region.

Of most relevance to this paper is the work of Permuter and Weissman in \cite{PermuterWiessman_vending_11}, relating to source coding with SI, also termed the Vending Machine (VM). A Wyner-Ziv model as in \cite{WynerZiv} was considered, yet with actions. Actions are performed at the encoder or the decoder and might affect the quality of the SI available to the decoder. This extension was proved to have a significant impact on the performance of such a system. In \cite{zhao14_compression_actions}, Zhao \emph{et al.} studied a new role for actions that affect the distribution of an information source. An action-dependent information source is generated and a reconstruction of the latter is required at terminal node. The optimal compression rate was characterized for the lossless case, and bounds on the rate-cost distortion function were given. Later on, in \cite{in_block}, Simeone considered  a VM model, but with sources that are not memoryless and with actions that might also be affected by causal observation of the SI. In \cite{action_double_sided}, Kittichokechai \emph{et al.} considered a source coding model where actions affect the generation of two-sided SI sequences; one is available to the encoder and the other one to the decoder. In \cite{information_embedding}, Ahmadi \emph{et al.} studied a new role of actions, where an additional decoder observes a function of the actions. A characterization of the information that this decoder can reconstruct was given for several setups. In \cite{ChiaAsnaniWeissman13_multi_terminal_source_coding_action}, Chia \emph{et al.} studied a multi-user setup of the VM; two decoders can observe different SI sequences, where both sequences were generated according to the same actions. In all the cited papers, actions were proved to be efficient while acquiring SI or generating an action-dependent information source; here, we study the role for actions that affect the distribution of an information source in a multi-user setups.

The remainder of the paper is organized as follows. In Section \ref{section:notation}, we formulate the problem for all communication models. Section \ref{section:main} includes  a statement of our main results regarding the optimal rate regions for case A, case B and the set of achievable rates for the general network scenario. Section \ref{section:examples} describes two binary examples. Section \ref{section:proofs} includes the proofs of case A and case B. A detailed proof for the network coding scenario is provided in Section \ref{section_network}. Finally, Section \ref{section:conclusions} summarizes the main achievements and insights presented in this work along with some possible future work.

\section{Notation And Problem Definition}\label{section:notation}
Let $\mathcal{X}$ be a finite set, and let $\mathcal{X}^{n}$ denote the set of all $n$-tuples of elements from $\mathcal{X}$. An element from $\mathcal{X}^{n}$ is denoted by $x^{n}=(x_1,x_2,\dots,x_n)$. If the dimension is clear from the context then boldface letters $\mathbf{x}$ will refer to $x^n$. Random variables are denoted by uppercase letters, $X$, and the previous notation also holds here, e.g. $X^{n}=(X_1,X_2,\dots,X_n)$ and $\vX$ stands for $X^{n}$. The probability mass function of $X$, the joint distribution function of $X$ and $Y$, and the conditional distribution of $X$ given $Y$ will be denoted by $P_{X}, P_{X,Y}$ and $P_{X|Y}$, respectively. Additionally, the notation $\ceil{x}$ stands for the smallest integer greater than $x$.

We consider a system of correlated sources with actions. Let us refer to case A as the case where the decoder is allowed to perform actions and to case B as the case where encoder $1$ performs actions. We provide here a definition for the setting of case A and the definition for the setting of case B is straightforward. The source sequence $\vX$ is such that $X_{i}\in \mathcal{X}$ for $i\in [1,n]$ and is distributed i.i.d. with a pmf $P_{X}$. The first encoder measures a sequence $X^{n}$ and encodes it in a message $T_{1}\in\{1,\dots,2^{nR_{X}}\}$, which is transmitted to the decoder. The decoder receives the index $T_{1}$ and selects an action sequence, where $A^{n}\in\mathcal{A}^{n}$. The action sequence affects the generation of the other source sequence $\vY$, which is the output of a discrete memoryless channel $P_{Y|X,A}$ with inputs of$(\vX,\vA)$. Specifically, given $\vX=\vx$ and $\vA=\va$, the source sequence $\vY$ is distributed as
\begin{align}\label{setup:memoryless}
p(y^{n}|x^{n},a^{n})=\prod_{i=1}^{n} p(y_{i}|x_{i},a_{i}).
\end{align}
Encoder $2$ receives the observation $y^{n}$ and encodes it in a message $T_{2}\in\{1,\dots,2^{nR_{Y}}\}$. The estimated sequences $(\hat{X}^n,\hat{Y}^n)$ are then obtained at the decoder as a function of the messages $T_{1}$ and $T_{2}$.

For the settings described above, a $(2^{nR_{X}},2^{nR_{Y}},n)$ \textit{code} for a block of length $n$ and rate pairs $(R_{X},R_{Y})$ consists of encoding functions:
\begin{align}\label{setup:encoders}
    T_{1} &: \mathcal{X}^n \rightarrow \{1,\dots,2^{nR_X}\}, \nonumber \\
    T_{2} &: \mathcal{Y}^n \rightarrow \{1,\dots,2^{nR_Y}\},
\end{align}

strategy functions:
\begin{align}\label{setup:strategyfunctions}
h_{d}:&\{1,\dots,2^{nR_X}\} \rightarrow \mathcal{A}^n \hspace{4mm}\text{for case A}, \nonumber\\
h_{e}:& \mathcal{X}^n \rightarrow \mathcal{A}^n \hspace{20mm}\text{for case B},
\end{align}

and a decoding function:
\begin{equation}\label{setup:decoder}
g: \{1,\dots,2^{nR_X}\} \times \{1,\dots,2^{nR_Y}\} \rightarrow \widehat{\mathcal{X}}^n \times \widehat{\mathcal{Y}}^n.
\end{equation}
Actions taken are subject to a cost constraint $\Gamma$, that is,
\begin{equation}\label{setup:cost}
E\left[\frac{1}{n}\sum_{i=1}^{n}\Lambda(A_{i})\right]\leq \Gamma.
\end{equation}
The \textit{probability of error} for a code $(2^{nR_{X}},2^{nR_{Y}},n)$ is defined as $P_{e}^{(n)}=\Pr((X^{n},Y^{n})\neq g(T_{1},T_{2}))$. For a given cost constraint $\Gamma$, a rate pair $(R_{X},R_{Y})$ is said to be \textit{achievable} if there exists a sequence of codes $(2^{nR_{X}},2^{nR_{Y}},n)$ such that $P_{e}^{(n)}\rightarrow 0$ as $n\rightarrow\infty$ and the cost constraint, \eqref{setup:cost}, is satisfied. The \textit{optimal rate region} is the convex closure of the set of achievable rate pairs. Let us denote the optimal rate regions as $\mathcal{R}_{A}$ and $\mathcal{R}_{B}$ for case A and case B, respectively.

\subsection{Network Model}
A network is represented as a directed, acyclic graph $\mathcal{G}=(\mathcal{V},\mathcal{E})$, where $\mathcal{V}$ is the set of network nodes and $\mathcal{E}$ is the set of links, such that information can be sent noiselessly from node $i$ to node $j$ if $(i,j)\in \mathcal{E}$. Each edge $l \in \mathcal{E}$ is associated with a nonnegative real number $c_{l}$, which represents its capacity in bits per unit time. We denote the origin node of a link $l$ as $o(l)$ and the destination of a link $l$ as $d(l)$.

We specify a \textit{network of correlated sources with actions} $(\mathcal{V},\mathcal{E},s_{1},s_{2},\tau)$ as follows. The source sequence $X^{n}$ is such that $X_{i}\in \mathcal{X}$ for $i\in [1,n]$ is i.i.d. with a pmf $P_{X}$. Based on its source observation $X^{n}$, node $s_{1}\in\mathcal{V}$ selects an action sequence $A^{n}\in\mathcal{A}^{n}$. The action sequence affects the generation of the other source sequence $Y^{n}$, which is the output of a discrete memoryless channel $P_{Y|X,A}$ with inputs of $(X^{n},A^{n})$. More specifically, given $X^{n}=x^{n}$ and $A^{n}=a^{n}$, the source sequence $Y^{n}$ is distributed as $p(y^{n}|x^{n},a^{n})=\prod_{i=1}^{n} p(y_{i}|x_{i},a_{i})$. The source sequence $Y^{n}$ is available at node $s_{2}\in \mathcal{V}\setminus \{s_{1}\}$. The source sequences $(X^{n},Y^{n})$ are demanded at a set of terminal nodes denoted as $\tau \in \mathcal{V}\setminus \{s_{1},s_{2}\}$. We assume that the source nodes $s_1,s_2$ have no incoming links and that each node $t\in\tau$ has no outgoing links.

For any vector of rates $(R_{l})_{l\in\mathcal{E}}$, a $\left(\left(2^{nR_{l}}\right)_{l\in\mathcal{E}},n\right)$ \emph{source code} consists of strategy function:
\begin{align}
     &h:\mathcal{X}^{n}\rightarrow\mathcal{A}^{n},
\end{align}
encoding functions:
\begin{align}
     &g_{l}: \mathcal{X}^{n}\rightarrow  \{1,\dots,2^{nR_{l}}\}  & \forall& l\in\mathcal{E},o(l)=s_{1},\nonumber\\
     &g_{l}: \mathcal{Y}^{n}\rightarrow  \{1,\dots,2^{nR_{l}}\}  & \forall& l\in\mathcal{E},o(l)=s_{2},\nonumber\\
     &g_{l}: \textstyle\prod_{l':d(l')=o(l)}\{1,\dots,2^{nR_{l'}}\}  \rightarrow  \{1,\dots,2^{nR_{l}}\} & \forall& l\in\mathcal{E},o(l)\not\in \{s_{1},s_{2}\},
\end{align}
and decoding functions, for each $t\in\tau$:
\begin{align}
     &\phi_t: \textstyle\prod_{l:d(l)=t}\{1,\dots,2^{nR_{l}}\}\rightarrow \hat{\mathcal{X}}^{n}\times\hat{\mathcal{Y}}^{n}.
\end{align}
We are interested in the set of possible values $(c_{l})_{l\in\mathcal{E}}$ for which for any $\epsilon>0$ there exists a sufficiently large $n$ and a $\left(\left(2^{nR_{l}}\right)_{l\in\mathcal{E}},n\right)$ code exists satisfying $R_{l}\leq c_{l}$ for all $l\in\mathcal{E}$, such that $\Pr((\hat{X}_t^{n},\hat{Y}_t^{n})\neq (X^{n},Y^{n}))\geq 1-\epsilon$ for each $t\in\tau$ and $E\left[\frac{1}{n}\sum_{i=1}^{n}\Lambda(A_{i})\right]\leq \Gamma$. We call the closure of this set of rate vectors the \textit{set of the achievable rates}, which we denote by $\mathcal{R}_{N}$.

Given any set $A\subset\mathcal{V}$ and a node $t\in\mathcal{V}\setminus A$, a \textit{cut} $\mathcal{V}_{A;t}$ is a subset of vertices that includes $A$ but is disjoint from $t$, that is, $A\subseteq\mathcal{V}_{A;t}$ and $\mathcal{V}_{A;t}\cap t = \emptyset$. Given a cut $\mathcal{V}_{A;t}$, the \textit{capacity of a cut} $\mathcal{C}(\mathcal{V}_{A;t})$ is the sum over all capacities of edges $l\in \mathcal{E}$ such that $o(l)\in\mathcal{V}_{A;t}$ and $d(l)\not\in\mathcal{V}_{A;t}$; that is,
\begin{equation}\label{capacityofacut}
    \mathcal{C}(\mathcal{V}_{A;t}) = \sum \limits_{l\in\mathcal{E}: o(l)\in \mathcal{V}_{A;t}, d(l)\not\in \mathcal{V}_{A;t}} c_{l}.
\end{equation}
For given sets $A$ and node $t$, let $\mathcal{V}^{\ast}_{A;t}$ be the \textit{minimum cut}, which is the cut minimizes the capacity of a cut among all cuts $\mathcal{V}_{A;t}$.
Finally, for given non-intersecting sets $A,\tau$ we define $\mathcal{C}(\mathcal{V}^{\ast}_{A;\tau})= \min_{t\in\tau} \mathcal{C}(\mathcal{V}^{\ast}_{A;t})$.
\section{Main Results}\label{section:main}
The following three theorems are the main results in this paper.
\begin{theorem}\label{theorem:dec}
The optimal rate region $\mathcal{R}_{A}$ for case A (See Fig. \ref{fig:dec}), i.e. correlated sources with actions taken at the decoder, is the closure of the set of triplets $(R_{X},R_{Y},\Gamma)$ such that
\begin{subequations}\label{region:dec}
\begin{align}
    R_{X} &\geq H(X|Y,A) + I(X;A),\label{deceq1}\\
    R_{Y} &\geq H(Y|X,A), \label{deceq2}\\
    R_{X}+R_{Y} &\geq H(X,Y|A) + I(X;A),\label{deceq3}
\end{align}
\end{subequations}
where the joint distribution of $(X,A,Y)$ is of the form:
  \begin{equation}\label{joint}
  P_{X,A,Y}=P_{X}P_{A|X}P_{Y|A,X},
\end{equation}
under which $E\left[\Lambda(A)\right]\leq \Gamma$.
\end{theorem}
\begin{theorem}\label{theorem:enc}
The optimal rate region $\mathcal{R}_{B}$ for case B (See Fig. \ref{fig:enc}), i.e. correlated sources with actions taken at the encoder, is the closure of the set of triplets $(R_{X},R_{Y},\Gamma)$ such that
\begin{subequations}
\begin{align}
    R_{X} &\geq H(X|Y,A) + I(X;A)-I(Y;A),\label{enceq1}\\
    R_{Y} &\geq H(Y|X,A), \label{enceq2}\\
    R_{X}+R_{Y} &\geq H(X,Y|A) + I(X;A),\label{enceq3}
\end{align}
\end{subequations}
where the joint distribution of $(X,A,Y)$ is of the form \eqref{joint}, under which $E\left[\Lambda(A)\right]\leq \Gamma$.
\end{theorem}

Note, for a fixed distribution of the form \eqref{joint}, $\mathcal{R}_{A}\subseteq \mathcal{R}_{B}$. In $\mathcal{R}_{B}$, $R_{X}$ has a looser constraint, reduced by a non-negative factor of $I(Y;A)$, while the sum-rate remains the same. In case A, actions' information should be described completely within the rate $R_{X}$ prior to the generation of $Y^{n}$. However, in case B the indices $T_{1}$, $T_{2}$ are transmitted independently. Thus, reduction of $R_X$ is by the maximum amount of actions' information that is implied from the index $T_2$, i.e. $I(Y;A)$. Moreover, representation of $\mathcal{R}_{A}$ and $\mathcal{R}_{B}$ by their corner points shows that $(R_X,R_Y)=\left(H(X),H(Y|X,A)\right)$ is a common corner point for both setups. Thus, for high rates of $R_X$ actions at different nodes of the system might have the same affect on the optimal rate regions.

The regions $\mathcal{R}_{A}$ and $\mathcal{R}_{A}$ reduce to those investigated in \cite{PermuterWiessman_vending_11} for the special case of allocating unlimited rate for $R_Y$, equivalently, having the source $Y$ available at the decoder. Having unlimited $R_Y$ implies that \eqref{deceq2}-\eqref{deceq3} and \eqref{enceq2}-\eqref{enceq3} are redundant. Thus, we only have a constraint on $R_X$. Theorem \ref{theorem:dec} is then reduced to the result of \cite[Sec.\rom{2}]{PermuterWiessman_vending_11} source coding with SI where actions are taken at the decoder, while Theorem \ref{theorem:enc} is reduced to the result of \cite[Sec.\rom{3}]{PermuterWiessman_vending_11} source coding with SI where actions are taken at the encoder. Another special case is when considering deterministic actions, that is, $A=a$; let us write the original optimal rate region of SW as $\mathcal{R}_{SW}(P_{X},P_{Y|X})$, with the explicit dependence on $P_{X}$ and $P_{Y|X}$. For this setting, both $\mathcal{R}_A$ and $\mathcal{R}_B$ reduce to $\mathcal{R}_{SW}(P_{X},P_{Y|X,A=a})$.

%For the case $X=A$, $\mathcal{R}_{B}$ is reduced to the original region of SW, while the region $\mathcal{R}_{A}$ is a smaller region compared to the original SW region. The difference between $\mathcal{R}_{A}$ and the region of SW is the constraint on $R_{X}$. For this case, the action $A$ is the source $X$ and a requirement for the decoder to perform the actions is equivalent to knowledge of the source $X$. Therefore, we have the constraint $R_{X}\geq H(X)$, while at the original SW the solution is $R^{SW}_{X}\geq H(X|Y)$.
\begin{theorem}\label{theorem:network}
Given a correlated sources with action network $(\mathcal{V},\mathcal{E},s_{1},s_{2},\tau,\Gamma)$ (See Fig. \ref{fig:network}), the set of achievable rates $\mathcal{R}_{N}$ is such that
\begin{subequations}\label{NC:theoremeq}
\begin{align}
     \mathcal{C}(\mathcal{V}^{\ast}_{s_{1};,\tau}) &\geq I(X;A)-I(Y;A) + H(X|Y,A), \label{nceq1}\\
     \mathcal{C}(\mathcal{V}^{\ast}_{s_{2};\tau}) &\geq H(Y|X,A),\label{nceq2}\\
     \mathcal{C}(\mathcal{V}^{\ast}_{s_{1},s_{2};\tau}) &\geq I(X;A) + H(X,Y|A),\label{nceq3}
\end{align}
\end{subequations}
where the joint distribution of $(X,A,Y)$ is of the form \eqref{joint}, under which $E\left[\Lambda(A)\right]\leq \Gamma$.
\end{theorem}

Note, the network investigated here is an extension of case B. The network setting is reduced to case B by substituting $\mathcal{V}=\{s_{1},s_{2},t\}$ and $\mathcal{E}=\{(s_{1},t),(s_{2},t)\}$. Therefore, the right hand side of \eqref{NC:theoremeq} coincides with the information measurements in Theorem \ref{theorem:enc}.

\section{Examples}\label{section:examples}
In this section, we study two binary examples and derive the optimal rate regions $\mathcal{R}_{A}$, $\mathcal{R}_{B}$. For comparison, we also study a special scenario for which actions are taken before the the first source $X^{n}$ is known, and actions play the role of time-sharing random variable. This special scenario may seem a degenerate setup, but can lead to some insights when considering an implementation of such a system with actions. The first example illustrates a scenario where actions taken at different nodes of the system cannot affect the set of achievable rates, while the second example demonstrates how taking actions at different nodes of the system improve significantly the optimal rate region under a cost regime.
\begin{example}
This binary example illustrates a sensors' measurements transmission; $X$ and $Y$ are two measurements known at different nodes of the system. The measurement $X$ is a coarse measurement which is binary and distributed uniformly, while the measurement $Y$ corresponds to fine or coarse measurement depends on the taken actions. A low-cost actions correspond to a fine measurement within the measured range, and high-cost actions correspond to a coarse measurement identical to the $X$ measurement. This cost implies that the number of fine measurements needs to be above some threshold. Our goal is to characterize the rates that are required in order to know both measurements at the decoder under a cost regime.

The example is illustrated in Fig. \ref{fig:Ex2}; consider a binary case where $\mathcal{X}=\mathcal{Y}=\mathcal{A}=\{0,1\}$, and $X\sim Bern(.5)$. Let $Y$ be an output of a clean channel if $A=0$, and the output of a noisy-channel with crossover probability $0.5$ if $A=1$. Actions can be taken at the decoder (switch $1$ is closed), at the encoder (switch $2$ is closed) or in the special case of actions taken before the source $X$ is known (switch $1$ and switch $2$ are open). We consider a cost function $\Lambda(A)=A$ that induces $P(A=1) \leq \Gamma$
\begin{figure}[h]
\centering
\begin{psfrags}
    \psfragscanon
    \psfrag{A}[][][1]{$X\sim Bern(.5)$}
    \psfrag{B}[][][1]{$X$}
    \psfrag{C}[][][1]{Encoder}
    \psfrag{D}[][][1]{Decoder}
    \psfrag{E}[][][1]{$R_{X}$}
    \psfrag{F}[l][][1]{$R_{Y}$}
    \psfrag{G}[][][1]{A}
    \psfrag{H}[][][.7]{Fine measurement}
    \psfrag{I}[][][.7]{Coarse measurement}
    \psfrag{J}[][][1]{$0$}
    \psfrag{K}[][][1]{$1$}
    \psfrag{L}[cb][][1]{$Y$}
    \psfrag{M}[][][1]{$\hat{X},\hat{Y}$}
    \psfrag{N}[][][.7]{$0.5$}
    \psfrag{O}[][][0.6]{$1-\delta$}
    \psfrag{Z}[][][0.7]{$A=0$}
    \psfrag{X}[][][0.7]{$A=1$}
    \psfrag{S}[l][][0.7]{switch $1$}
    \psfrag{W}[l][][0.7]{switch $2$}
\includegraphics[height=6cm]{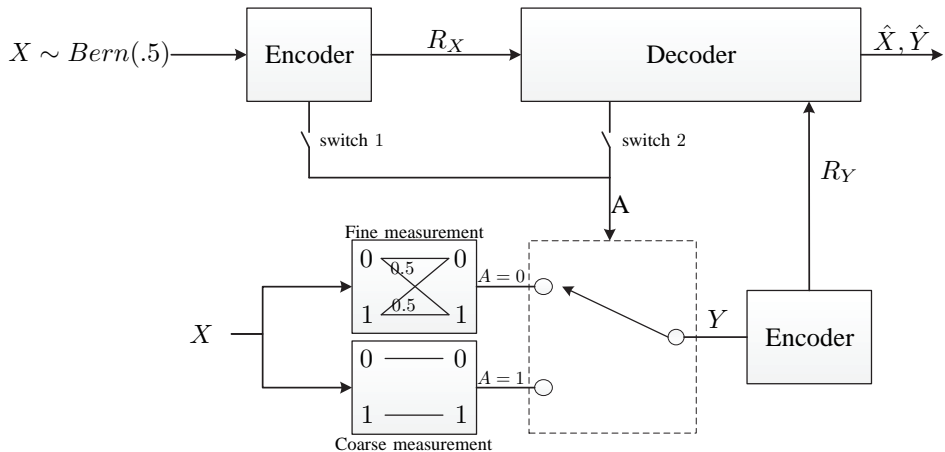}
\psfragscanoff
\end{psfrags}
\caption{The setup for example $1$. Actions can be performed by the decoder (switch $2$ is closed), by the encoder (switch $1$ is closed) or before $X$ is known (switch $1$ and switch $2$ are open). The switch in the dashed box corresponds to actions' performance.}
\label{fig:Ex2}
\end{figure}

\begin{itemize}
\item
Case A - actions are taken at the decoder; the setup is depicted in Fig. \ref{fig:Ex2}, with switch $1$ closed. A general conditional distribution connecting $X$ and $A$ is considered, with $P_{A|X}(1|0)=\alpha$ and $P_{A|X}(0|1)=\beta$. The optimal rate region, $\mathcal{R}_{A}$, is as follows:
\begin{align}
    R_{X} &\geq 1 - 0.5(\alpha+\bar{\beta})H_{b} (\frac{\alpha}{\alpha+\bar{\beta}}),\nonumber\\
    R_{Y} &\geq  0.5(\bar{\alpha}+\beta),\nonumber\\
    R_{X}+R_{Y} &\geq 1 + 0.5(\bar{\alpha}+\beta),
\end{align}
for some $\alpha,\beta \in[0,1]$ such that $0.5(\alpha+\bar{\beta}) \leq \Gamma$.
\item Case B - actions are taken at the encoder; the setup is depicted in Fig. \ref{fig:Ex2}, with switch $2$ closed. Calculating $\mathcal{R}_{B}$ with the same pmf as in the previous case yields:
\begin{align}
    R_{X} &\geq 1 - H_{b}(0.5\alpha + 0.25[\beta + \bar{\alpha}]) + 0.5(\bar{\alpha}+\beta),\nonumber\\
    R_{Y} &\geq  0.5(\bar{\alpha}+\beta),\nonumber\\
    R_{X}+R_{Y} &\geq 1 + 0.5(\bar{\alpha}+\beta),
\end{align}
for some $\alpha,\beta \in[0,1]$ such that $0.5(\alpha+\bar{\beta}) \leq \Gamma$.
\item Case C - Actions are taken before the source $X^{n}$ is known - for this case, actions contain no information of the source $X^{n}$ and play the role of a time-sharing random variable available to the system. Definitions of the probability of error, achievable rate pair and the optimal rate region, denoted by $\mathcal{R}_{A\bot X}$, remain as in the previous cases. For this scenario, it can be shown that the optimal rate region $\mathcal{R}_{A\bot X}$ is the set of $(R_{X},R_{Y},\Gamma)$ such that:
\begin{align}\label{RegionC}
    R_{X} \geq H(X|Y,A)\nonumber,\\
    R_{Y} \geq H(Y|X,A)\nonumber,\\
    R_{X}+R_{Y} \geq H(X,Y|A),
\end{align}
for some joint distribution $P_{X,A,Y}=P_{X}P_{A}P_{Y|A,X}$, under which $E\left[\Lambda(A)\right]\leq \Gamma$.

The setup is depicted in Fig. \ref{fig:Ex2} where both switches are open; we assume that $X\sim Bern(\bar{\alpha})$ and the optimal rate region, $\mathcal{R}_{A\bot X}$, is as follows:
\begin{align}
    R_{X} &\geq \bar{\alpha}, \nonumber\\
    R_{Y} &\geq \bar{\alpha}, \nonumber \\
    R_{X}+R_{Y} &\geq 1 + \bar{\alpha},
\end{align} for some $\alpha \geq \Gamma$.
\end{itemize}

Remarkably, the unions over the three regions coincide for any value of $\Gamma$. Let us provide the coding scheme for minimizing the regions; substitute $\alpha=\Gamma$ and $\bar{\beta}=\Gamma$ (which satisfies the cost constraint) so that all the three regions are then minimized. The minimized region for three cases as a function of the cost, $\Gamma$, is then:
\begin{align}
    R_{X}& \geq 1-\Gamma, \nonumber\\
    R_{Y}& \geq 1-\Gamma, \nonumber\\
    R_{X}+R_{Y}& \geq 2-\Gamma.
\end{align}

This equivalence can happen in systems for which greedy policy is optimal. A greedy policy is associated with a system for which different observations of $X$ lead to the same actions strategy. For instance, in example $2$ greedy policy implies $A=1$ which yields more correlation between $X$ and $Y$ and thus a greater achievable region. Note that this policy has no dependence on the source $X$, and we are constrained only by the cost $\Gamma$.
\end{example}
\begin{example}
The example is depicted in Fig. \ref{fig:Ex1}; we consider the previous example but with a different channel characterization of the source $Y$. Let $Y$ be an output of a Z-channel with crossover probability $\delta$ if $A=0$, and the output of an S-channel with crossover probability $\delta$ if $A=1$. Again, actions can be taken at the decoder (switch $1$ is closed), at the encoder (switch $2$ is closed) or in the case that actions are taken before the source $X$ is known (switch $1$ and switch $2$ are open). We consider a cost function $\Lambda(A)=A$ which induces $P(A=1) \leq \Gamma$.
\begin{figure}[h!]
\centering
\begin{psfrags}
    \psfragscanon
    \psfrag{A}[][][1]{$X\sim Bern(.5)$}
    \psfrag{B}[][][1]{$X$}
    \psfrag{C}[][][1]{Encoder}
    \psfrag{D}[][][1]{Decoder}
    \psfrag{E}[][][1]{$R_{X}$}
    \psfrag{F}[l][][1]{$R_{Y}$}
    \psfrag{G}[][][1]{A}
    \psfrag{H}[][][.9]{Z-Channel}
    \psfrag{I}[][][.9]{S-Channel}
    \psfrag{J}[][][1]{$0$}
    \psfrag{K}[][][1]{$1$}
    \psfrag{L}[cb][][1]{$Y$}
    \psfrag{M}[][][1]{$\hat{X},\hat{Y}$}
    \psfrag{N}[][][1]{$\delta$}
    \psfrag{O}[][][0.6]{$1-\delta$}
    \psfrag{Z}[][][0.7]{$A=0$}
    \psfrag{X}[][][0.7]{$A=1$}
    \psfrag{S}[l][][0.7]{switch $1$}
    \psfrag{W}[l][][0.7]{switch $2$}
\includegraphics[height=6cm]{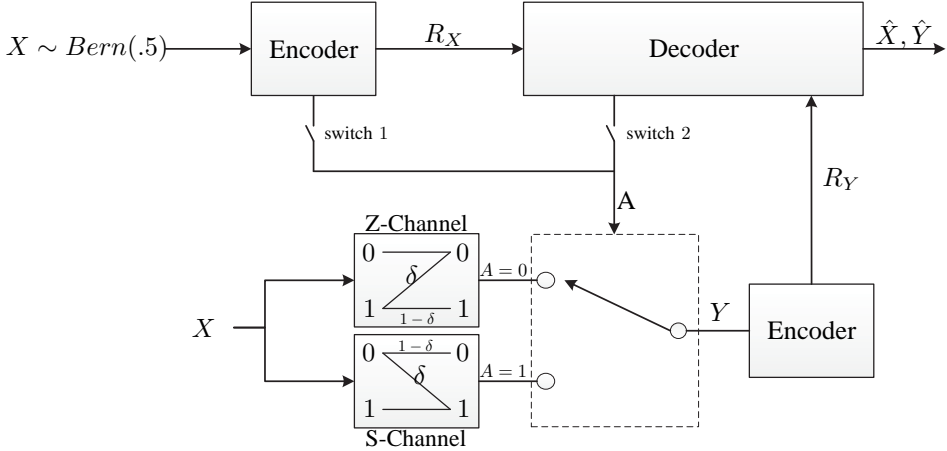}
\psfragscanoff
\end{psfrags}
\caption{The setup for example $2$. Actions can be taken at the decoder (switch $2$ is closed ), at the encoder (switch $1$ is closed) or before $X$ is known (switch $1$ and switch $2$ are open). The switch in the dashed box corresponds to actions' performance.}
\label{fig:Ex1}
\end{figure}

\begin{itemize}
\item
Case A - actions are taken at the decoder; the setup is depicted in Fig. \ref{fig:Ex1} for the case that switch $2$ is closed. A general conditional distribution connecting $X$ and $A$ is considered, with $P_{A|X}(1|0)=\alpha$ and $P_{A|X}(0|1)=\beta$. The optimal rate region, $\mathcal{R}_{A}$, is as follows:

\begin{align}
    R_{X} &\geq 1 - 0.5(\alpha+\bar{\beta})H_{b}(\frac{\bar{\beta}}{\alpha+\bar{\beta}})- 0.5(\beta+\bar{\alpha})H_{b}(\frac{\bar{\alpha}}{\beta+\bar{\alpha}}){\nonumber}\\& \hspace{3mm}+ 0.5(\bar{\alpha}+\beta\delta)H_{b}(\frac{\bar{\alpha}}{\bar{\alpha}+\beta\delta}) + 0.5(\bar{\beta}+\alpha\delta)H_{b}(\frac{\bar{\beta}}{\bar{\beta}+\alpha\delta}),\nonumber\\
    R_{Y} &\geq  0.5(\alpha+\beta)H_{b}(\delta), \nonumber\\
    R_{X}+R_{Y} &\geq 1 + 0.5(\alpha+\beta)H_{b}(\delta),
\end{align}
for some $\alpha,\beta \in[0,1]$ such that $0.5(\alpha+\bar{\beta}) \leq \Gamma$ and $\bar{\alpha}$ stands for $1-\alpha$.

\item Case B - actions are taken at the encoder; the setup is depicted in Fig. \ref{fig:Ex1} for the case that switch $1$ is closed. A conditional distribution is assumed as in case $A$. The optimal rate region ,$\mathcal{R}_{B}$, for this case is as follows:
\begin{align}
    R_{X} &\geq 1 + 0.5(\alpha+\beta)H_{b}(\delta) - H_{b}(0.5[1 + \alpha\delta - \beta\delta]),\nonumber\\
    R_{Y} &\geq 0.5(\alpha+\beta)H_{b}(\delta),\nonumber\\
    R_{X}+R_{Y} &\geq 1 + 0.5(\alpha+\beta)H_{b}(\delta),
\end{align}
for some $\alpha,\beta \in[0,1]$ such that $0.5(\alpha+\bar{\beta}) \leq \Gamma$.

Note, the optimal rate region $\mathcal{R}_{B}$ is minimized by taking $A=X$ for the case of $\Gamma\geq0.5$.
\item Case C - actions are taken before the source $X^{n}$ is known; the setup is depicted in Fig. \ref{fig:Ex1} where both switches are open. The optimal rate region, $\mathcal{R}_{A\bot X}$, for example $2$ is:
\begin{align}
    R_{X} &\geq 0.5(1+\delta)H_{b}(\frac{1}{1+\delta}),\nonumber\\
    R_{Y} &\geq 0.5H_{b}(\delta),\nonumber\\
    R_{X}+R_{Y} &\geq 1 + 0.5 H_{b}(\delta).
\end{align}
\end{itemize}
Note that the region is independent of $\alpha$ and no union is needed here. This fact implies that $\mathcal{R}_{A\bot X}$ is also independent of the cost $\Gamma$ and only depends on the value of $\delta$.

To gain some intuition regarding the optimal rate regions, we draw the results for $\Gamma=0.3$ and $\delta=0.5$ in Fig. \ref{fig:Results_Ex}. Let us examine the curved dashed blue line, which corresponds to case A; its corner point coincides with the black line (squared-marker) and tends to the red line (triangled-marker) in different parts of the region. For a high $R_{X}$, an action is transmitted explicitly within $R_X$ and induces high correlation with the source $X$. Decreasing $R_{X}$ implies that $P_{A|X}$ induces the action to be less correlative with $X$ and, therefore, tends to the region $\mathcal{R}_{A\bot X}$. Nevertheless, the blue plot achieves better performance in $R_X$ than the red plot, which implies that correlation is required to achieve minimum $R_X$. Clearly, case A and case B have greater optimal region than the case of actions independent of $X^{n}$, thus time-sharing is not optimal when investigating an action-dependent system.

\begin{figure}[h!]
\centering
\begin{psfrags}
    \psfragscanon
    \psfrag{X}[ct][][1]{$R_{X}[Bits/Symbol]$}
    \psfrag{Y}[cb][][1]{$R_{Y}[Bits/Symbol]$}
    \psfrag{B}[l][][.9]{Case C - Actions are taken before $X^{n}$ is known}
    \psfrag{D}[l][][.9]{Case A - Actions are taken at the decoder}
    \psfrag{E}[l][][.9]{Case B - Actions are taken at the encoder}
    \psfrag{Q}[c][][0.8]{$\Gamma=0.5,\delta=0.5$}
    \psfrag{W}[c][][0.8]{$\Gamma=0.5,\delta=0.3$}
    \psfrag{T}[c][][1.1]{Comparison for $\Gamma=0.3,\delta=0.5$}
    \psfrag{U}[c][][0.8]{$\Gamma=0.3,\delta=0.3$}
\includegraphics[width=16cm]{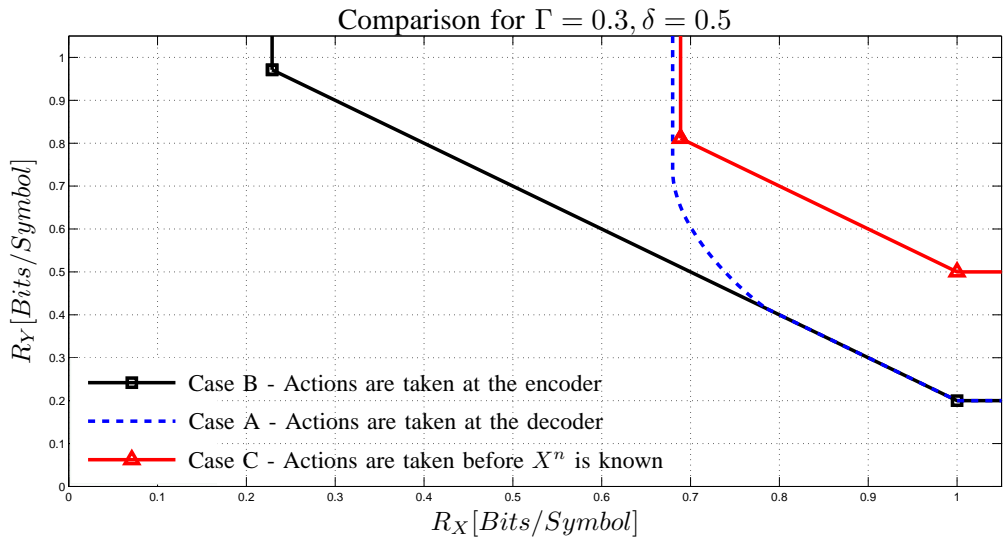}
\psfragscanoff
\end{psfrags}
\caption{The optimal rate regions for three cases of Example $2$.}
\label{fig:Results_Ex}
\end{figure}
\end{example}
\section{Proofs of Case A and Case B}\label{section:proofs}
In this section we present the proof of Theorem \ref{theorem:dec} and Theorem \ref{theorem:enc}. As mentioned in Section \ref{section:main}, the network model that was studied in Theorem \ref{theorem:network} can be reduced to case B under certain conditions. Thus, the converse of Theorem \ref{theorem:enc} is omitted here and can be followed directly from the converse of Theorem \ref{theorem:network}, which is provided in Section \ref{section_network}. However, we provide an alternative achievability proof, which is less complicated than the direct method of Theorem \ref{theorem:network}.
\subsection{Proof of Theorem \ref{theorem:dec}}\label{proof:dec}
\textbf{Sketch of Achievability:}
At the first stage, the identity of the action sequence is transmitted from encoder $1$ to the decoder; generate a codebook of actions containing $2^{nI(X;A)}$ independent codewords, where each codeword is generated according to $P_{A}$. Encoder $1$ looks in the codebook for a codeword which is jointly typical with the source observation $x^n$, and transmits this codeword to the decoder using a rate of $I(X;A)$.
Note that the optimal rate region, \eqref{region:dec},can be written as:
\begin{align}\label{proof_dec_alternative}
    R_{X} - I(X;A) &\geq H(X|Y,A), \nonumber\\
    R_{Y} &\geq H(Y|X,A), \nonumber\\
    \left[R_{X} - I(X;A)\right] + R_{Y} &\geq H(X,Y|A).
\end{align}
Before proceeding to the last step of the proof, note that the triplet $(A^{n},X^{n},Y^{n})$ is jointly typical with high probability. The source $Y^{n}$ is an output of a memoryless channel that is conditioned on the pair $(X^{n},A^{n})$; this pair is jointly typical with high probability according to the covering lemma \cite[Chapter 3]{ElGamal}. Now, using the fact that the triplet is jointly typical, the right hand side of \eqref{proof_dec_alternative} is achieved by implementing a SW coding scheme, where actions are treated as SI available at the decoder.

\textbf{Converse:}

Assume that a sequence $(2^{nR_{X}},2^{nR_{Y}},n)$ of achievable codes exists.
For the rate that is used by encoder $1$, consider:
 \begin{align}
   nR_{X} &\geq H(T_{1}) \nonumber\\
    &\stackrel{(a)}= H(T_{1})+H(A^{n}|T_{1}) + H(X^{n}|Y^{n},T_{1}) - H(X^{n}|Y^{n},T_{1}) \nonumber\\
    &\stackrel{(b)}\geq H(A^{n})+H(T_{1}|A^{n})+H(X^{n}|Y^{n},T_{1})- n\epsilon_{n} \nonumber\\
    &\stackrel{(c)}\geq H(A^{n})+H(X^{n},T_{1}|A^{n},Y^{n})- n\epsilon_{n}\nonumber\\
    &= H(A^{n}) + H(X^{n}|A^{n},Y^{n}) + H(T_{1}|X^{n},A^{n},Y^{n})- n\epsilon_{n}\nonumber\\
    &\stackrel{(d)}= H(A^{n}) + H(X^{n}|A^{n},Y^{n})- n\epsilon_{n}\nonumber\\
    &\stackrel{(e)}= H(A^{n})-H(A^{n}|X^{n})+H(X^{n}|A^{n},Y^{n})- n\epsilon_{n}\nonumber\\
    &\stackrel{(f)}= H(X^{n})-H(Y^{n}|A^{n})+H(Y^{n}|A^{n},X^{n})- n\epsilon_{n}\nonumber\\
    &\stackrel{(g)}= \sum_{i=1}^{n} \left[ H(X_{i})-H(Y_{i}|Y^{i-1},A^{n}) +H(Y_{i}|A_{i},X_{i})\right]- n\epsilon_{n}\nonumber\\
    &\stackrel{(h)}\geq \sum_{i=1}^{n} \left[H(X_{i})-H(Y_{i}|A_{i})+H(Y_{i}|A_{i},X_{i})\right]- n\epsilon_{n}\nonumber\\
    &\stackrel{(i)}\geq \sum_{i=1}^{n} \left[I(X_{i};A_{i})+H(X_{i}|A_{i},Y_{i})\right]- n\epsilon_{n}, \label{eq_converse_dec1}
 \end{align}
where:
\begin{itemize}
  \item [(a)] follows from the fact that $A^{n}$ is a deterministic function of the index $T_{1}$;
  \item [(b)] follows from Fano's inequality and properties of joint entropy;
  \item [(c)] follows from the fact that conditioning reduces entropy;
  \item [(d)] follows from the fact that $T_{1}$ is a deterministic function of $X^{n}$;
  \item [(e)] follows from the fact that $A^{n}$ is a deterministic function of $X^{n}$;
  \item [(f)] follows from the properties of mutual information;
  \item [(g)] follows from the fact that $X^{n}$ is i.i.d. and the memoryless property \eqref{setup:memoryless};
  \item [(h)] follows from the fact that conditioning reduces entropy;
  \item [(i)] follows from the properties of mutual information.
\end{itemize}

For the rate that is used by encoder $2$:
 \begin{align}
   nR_{Y} &\geq H(T_{2}) \nonumber\\
    &\stackrel{(a)}\geq H(T_{2},Y^{n}|X^{n})-H(Y^{n}|T_{2},X^{n}) \nonumber\\
    &\stackrel{(b)}\geq H(T_{2},Y^{n}|X^{n}) -n\epsilon_{n}\nonumber\\
    &\stackrel{(c)}= H(Y^{n}|X^{n}) - n\epsilon_{n} \nonumber\\
    &\stackrel{(d)}= H(Y^{n}|X^{n},A^{n}) -n\epsilon_{n} \nonumber\\
    &\stackrel{(e)}= \sum_{i=1}^{n} \left[H(Y_{i}|A_{i},X_{i})\right] -n\epsilon_{n}, \label{eq_converse_dec2}
 \end{align}
 where:
\begin{itemize}
  \item [(a)] follows from the fact that conditioning reduces entropy;
  \item [(b)] follows from Fano's inequality;
  \item [(c)] follows from the fact that $T_{2}$ is a deterministic functions of $Y^{n}$;
  \item [(d)] follows from the fact that $A^{n}$ is deterministic functions of $X^{n}$;
  \item [(e)] follows from the memoryless property \eqref{setup:memoryless}.
\end{itemize}

The last converse is for the sum-rate of the encoders:
 \begin{align}
    n(R_{X}+R_{Y}) &\geq H(T_{1},T_{2}) \nonumber\\
    &= H(T_{1},T_{2},X^{n},Y^{n})-H(X^{n},Y^{n}|T_{1},T_{2}) \nonumber\\
    &\stackrel{(a)}\geq H(X^{n},Y^{n})+H(T_{1},T_{2}|X^{n},Y^{n})-n\epsilon_{n}\nonumber\\
    &\stackrel{(b)}= H(X^{n},Y^{n})-n\epsilon_{n}\nonumber\\
    &\stackrel{(c)}= H(X^{n})+H(Y^{n}|X^{n},A^{n})-n\epsilon_{n}\nonumber\\
    &\stackrel{(d)}= \sum_{i=1}^{n} \left[H(X_{i})+H(Y_{i}|X_{i},A_{i})\right] -n\epsilon_{n}\nonumber\\
    &= \sum_{i=1}^{n} \left[I(X_{i};A_{i})+H(X_{i},Y_{i}|A_{i})\right] -n\epsilon_{n}
    \end{align}
where:
\begin{itemize}
  \item [(a)] follows from Fano's inequality and the properties of joint entropy;
  \item [(b)] follows from the fact that $T_{1}$ and $T_{2}$ are deterministic functions of $X^{n}$ and $Y^{n}$, respectively;
  \item [(c)] follows from the fact that $A^{n}$ is a deterministic function of $X^{n}$;
  \item [(d)] follows from the fact that $X^{n}$ is memoryless and the memoryless property \ref{setup:memoryless}.
\end{itemize}
Derivation of the single letter terms is by using a standard time-sharing techinque.
Thus, we have shown the bounds:
 \begin{align}
    R_{X} &\geq I(X;A) + H(X|A,Y)- \epsilon_{n},\nonumber\\
    R_{Y} &\geq H(Y|A,X) -\epsilon_{n},\nonumber\\
    R_{X} + R_{y} &\geq I(X;A) + H(X,Y|A) -\epsilon_{n}.
 \end{align}
The proof is completed by taking $n\rightarrow\infty$, which implies $\epsilon_{n}\rightarrow 0$ since $(R_X,R_Y)$ are achievable.
\subsection{Achievability of Theorem \ref{theorem:enc}}\label{proof:enc}
The achievability proof is based on arguments of time sharing; namely, we prove the corner points of $\mathcal{R}_{B}$ to be achievable and conclude that the convex region is also achievable. Throughout the proof, we differentiate two cases according to the sign of the term  $I(X;A)-I(Y;A)$. The corner points of $\mathcal{R}_{B}$ are illustrated in Fig. \ref{fig:region_enc}, and can be written as:
\begin{align}
(R_X,R_Y) &= \left(I(X;A)-I(Y;A) + H(X|Y,A),H(Y)\right), \label{eq:enc_cp1}\\
(R_X,R_Y) &=  \left(H(X),H(Y|X,A)\right).\label{eq:enc_cp2}
\end{align}
\begin{figure}[h!]
\centering
\begin{psfrags}
    \psfragscanon
    \psfrag{A}[c][][.9]{$H(Y)$}
    \psfrag{B}[c][][.9]{$H(Y|X,A)$}
    \psfrag{C}[c][][.9]{$I(X;A)-I(Y;A)$}
    \psfrag{D}[c][][.9]{$+H(X|Y,A)$}
    \psfrag{E}[c][][0.8]{$H(X)$}
    \psfrag{X}[c][][1.1]{$R_{X}$}
    \psfrag{Y}[c][][1.1]{$R_{Y}$}
\includegraphics[width=8cm]{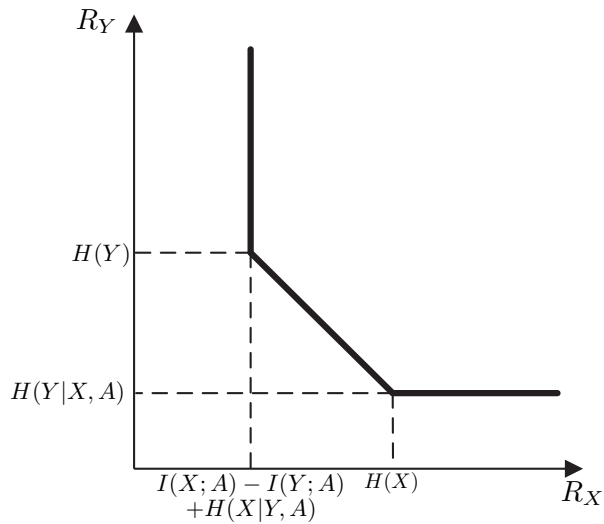}
\psfragscanoff
\end{psfrags}
\caption{The optimal rate region, $\mathcal{R}_B$, for case B.}
\label{fig:region_enc}
\end{figure}

The corner point in \eqref{eq:enc_cp1} can be achieved as follows; we first transmit the source sequence $Y^{n}$ in a lossless manner at a rate of $H(Y)$ to the decoder, then our problem reduces to that of \cite[Sec.\rom{3}]{PermuterWiessman_vending_11}-source coding with SI where actions are taken at the encoder. The proof for the rate $R_{X} = I(X;A)-I(Y;A) + H(X|Y,A)$ is omitted here, and can be found in \cite[Sec.\rom{3}]{PermuterWiessman_vending_11}.

%For the case $I(X;A)-I(Y;A)\geq 0$, the proof is similar to that given for the Wyner-Ziv setup \cite{WynerZiv}. Actions can be considered as a distorted version of the source $X$, and the source $Y$ is the SI available to the decoder. Thus, transmitting the action sequences to the decoder is by using a rate of $I(X;A)-I(Y;A)$. Note, the Markov chain that was required in \cite{WynerZiv} is not needed here, since the triplet $(X^{n},Y^{n},A^{n})$ is jointly typical with high probability, on the basis of the arguments that were given in the proof of Theorem \ref{theorem:dec}. Finally, we implement a random binning scheme at rate $H(X|Y,A)$ to transmit the source $X^{n}$ to the decoder, with $(Y^{n},A^{n})$ that is considered as SI available at the decoder.
%
%For the case $I(Y;A)-I(X;A)\geq 0$, we use a Gelfand Pinsker coding scheme \cite{GePi80}. Actions are considered as an input to a channel with state, $P_{Y|X,A}$, where the output of the channel is $Y$ and the channel state is $X$. Since $Y$ is known at the decoder, we use the capacity of this channel i.e. $I(Y;A)-I(X;A)$ to transmit part of the source $X$ at this rate. We complete the proof by implementing a random binning scheme at a rate of $H(X|Y,A)-(I(X;A)-I(Y;A))$. It is important to mention that the Markov chain that was required in the original coding scheme is not required here since the triplet $(X^{n},Y^{n},A^{n})$ is jointly typical with high probability.

The corner point in \eqref{eq:enc_cp2} is, indeed, the common corner point for case A and case B as mentioned in Section \ref{section:main}. The rate $R_X$ in \eqref{eq:enc_cp2} can be written also as $H(X)$; thus, this rate is used to transmit the source $X^{n}$ in a lossless manner to the decoder. Having received the source $X^{n}$, the decoder obtains $A^{n}$, which is a deterministic function of $X^{n}$. Later, a trivial source coding scheme for the source $Y^{n}$ is used at a rate of $H(Y|X,A)$, where $(X^{n},A^{n})$ are considered as SI available to the decoder.

\section{Proof of Theorem \ref{theorem:network}}\label{section_network}
In this section, a detailed proof for Theorem \ref{theorem:network} is provided. The code construction, encoding and the decoding procedures are presented in Subsection \ref{subsec:direct}, while the analysis of the probability of error will be given in Appendix \ref{app:analysis}. In Subsection \ref{subsec:lemma}, Lemma \ref{lemma:bounds} states an upper bound on the probability of error that two different inputs to a randomized network yield the same output, followed by a multicast example and the proof of the lemma. Finally, the proof of the converse for Theorem \ref{theorem:network} is given in Subsection \ref{subsec:converse}.
%The result of Lemma \ref{lemma:bounds} is a generalization of the result in \cite[Appendix A]{random}, and can be used in a broader class of network coding problems, as demonstrated in Appendix \ref{appendix:proof_lemma_1}.
\subsection{Direct}\label{subsec:direct}
The direct part is based on RLNC in the finite field $\mathbb{F}_{2^{n}}$. Construction of the code comprises codebook generation of the actions codewords and, later on, random binning of the source sequences $\vX$ and $\vY$. The bins and the action codewords will then be the input to the network, but after representing each input as a vector of elements from $\mathbb{F}_{2^{n}}$. For the transmission in the network, we rely on the scalar algebraic approach introduced by Koetter and Medard \cite{scalar} and represent the linear mapping from inputs to the output in a terminal node as a matrix. Regarding the decoding procedure, in \cite{random} decoding was based on min-entropy or maximum a posteriori probability procedures. However, in our setup, the triplet $(\vX,\vA,\vY)$ is not distributed i.i.d. since actions are a function of the complete source sequence $\vX$; therefore, we adopt a strong typicality decoding procedure.

Throughout the direct proof, differentiation between two cases is based on the sign of the term $I(X;A)-I(Y;A)$. Since actions are functions of $X^{n}$, the generation rate of actions that is required to preserve joint typicality of the triplet $(\vX,\vA,\vY)$ is $I(X;A)$. When the sign of $I(X;A)-I(Y;A)$ is positive, we generate actions at a rate of $I(X;A)$ and choose the actions' sequence according to a joint typicality criteria with the observation of $X^{n}$. Then the term $I(X;A)-I(Y;A)$ corresponds to the rate that is required in order describe actions' sequence to a node that has access to $\vY$.
For a negative sign, $I(Y;A) - I(X;A)\geq 0$, actions contain more information of the source $Y$ than the source $X$. We exploit this fact by generating actions at rate of $I(Y;A)$, which is greater than the required generation rate, and randomly bin them at a rate of $I(Y;A)-I(X;A)$. It then follows that any node which has access to $Y^{n}$ can decode the actions, and thus finding the bin that contains the actions. The bin index is considered as a message, which is used to decrease the required minimum cut, $\mathcal{C}(\mathcal{V}^{\ast}_{1;t})$, and improve the achievable region.

\textbf{The case $I(X;A)-I(Y;A)\geq 0$:}
Fix a joint distribution of $P_{X,A,Y}=P_{X}P_{A|X}P_{Y|A,X}$, where the source distribution $P_{X}$ and $P_{Y|A,X}$ are given.

\underline{Code construction:}
\begin{itemize}
\item The $\vX$ sequences are randomly binned into $2^{nr_1}$ bins, where $r_1 \triangleq H(X)+\epsilon$, for some $\epsilon > 0$. Each bin can be represented as $nr_1$ bits, or alternatively as a vector of $\ceil{r_1}$ elements from the finite field $\mathbb{F}_{2^{n}}$. The bin vector of $X^{n}$ will be denoted as $\IX$, consisting of $\ceil{r_1}$ elements. The $\vY$ sequences are randomly binned into $2^{nr_2}$ bins, where $r_2 \triangleq H(Y)+\epsilon$. Again, the bin vector of the sequence $\vY$ will be denoted by $\IY$, consisting of $\ceil{r_2}$ elements from $\mathbb{F}_{2^{n}}$. The bin vectors $\IX$ and $\IY$ will be part of the input to the network.
\item A codebook $\mathcal{C}$ of actions codewords is generated, consisting of $2^{nr_A}$ independent codewords, $\vA(i)$, $i\in\{1,2,\dots,2^{nr_A}\}$, where each codeword is distributed i.i.d. according to $\sim\prod_{j=1}^{n} P_{A}(a_j)$. Each codeword $\vA$ is represented by a vector of elements from $\mathbb{F}_{2^{n}}$, denoted by $\IA$ and consisting of $\ceil{r_A}$ elements.
\item The inputs to the network will be the source bins $\IX$, $\IY$, and actions codewords $\IA$, each consisting of elements in $\mathbb{F}_{2^{n}}$. Each element in the input vectors $\IX$, $\IY$, and $\IA$ is denoted by $U_i$, where $i\in\{1,\dots,\ceil{r_1}+\ceil{r_2}+\ceil{r_A}\}$. Let $o(U_i)$ be equal to $s_1$ if $U_i$ is an element in the vector $\IX$ or $\IA$, and $o(U_i)=s_2$ if $U_i$ is an element in the vector $\IY$.
\end{itemize}

  The information process $V_{j}$ transmitted on a link $j\in\mathcal{E}$ is formed as a linear combination, in $\mathbb{F}_{2^{n}}$, of link $j$'s inputs, i.e. source elements, $U_{i}$, for which $o(U_i)=o(j)$ and input processes $V_{l}$ for which $d(l)=o(j)$. This can be represented by the equation
\begin{equation}
    V_{j}= \sum_{i: o(U_i)=o(j)} b_{i,j} U_{i} + \sum_{l:d(l)=o(j)} f_{l,j} V_{l}.
\end{equation}
The coefficients $\{b_{i,j},f_{l,j}\}$ are generated uniformly from the finite field $\mathbb{F}_{2^{n}}$ and collected into matrices $\mathbf{B}=\{b_{i,j}\}$ and $\mathbf{F}=\{f_{l,j}\}$; note the dimensions $|\mathbf{B}|=(\ceil{r_1}+\ceil{r_2}+\ceil{r_A})\times |\mathcal{E}|$, and $|\mathbf{F}|=|\mathcal{E}|\times|\mathcal{E}|$. For acyclic graphs, we can assume that there exists an ancestral indexing of the links in $\mathcal{E}$. It then follows that the matrix $\mathbf{F}$ is upper triangular with zeros on the diagonal and there exists the inverse of $(\mathbf{I}-\mathbf{F})$, denoted by $\mathbf{G} \triangleq (\mathbf{I}-\mathbf{F})^{-1}$. Let $\mathbf{G}_v$ denote the sub-matrix consisting of only the columns of $\mathbf{G}$ corresponding to the input links of node $v$. Now, we can write the complete mapping from the input vector of the network, e.g. $\underline{U}=[\IX, \IA, \IY]$, to the input processes of some terminal node $t$ as:
\begin{equation}
    \underline{Z}_{t} = [\IX, \IA, \IY]\net,
\end{equation}
where $\underline{Z}_{t}$ is a vector consisting of the processes $V_{j}$ satisfying $d(j)=\{t\}$.

\underline{Encoding:}
Given the source realization $\vx$, node $s_1$ looks in the codebook for an index $i$ such that $\vA(i)$ is jointly typical with $\vx$; if there is none it outputs $i=1$. If there is more than one index, $i$ is set to the smallest among them. The source $\vY$ is then generated and available at node $s_2$. The input to the network will then be the vector $[\Ix, \Ia, \Iy ]$, where $\Ix,\Iy$ are the bins' sources, and $\Ia$ is the chosen actions codeword.

\underline{Decoding:}
Having received the vector $\underline{Z}_{t}$, each node $t\in\tau$ looks for a unique triplet $(\vX,\vA,\vY)\in \strtyp(X,A,Y)$ satisfying $[\IX,\IA,\IY]\net=\underline{Z}_{t}$.
%The notations $\mathbf{I}_{\hat{X}}, \mathbf{I}_{\hat{A}}, \mathbf{I}_{\hat{Y}}$ stand for the bin containing $\widehat{X}^{n},\widehat{A}^{n},\widehat{Y}^{n}$, respectively.

\textbf{The case $I(X;A) - I(Y;A)\leq 0$:}
Fix a joint distribution of $P_{X,A,Y}=P_{X}P_{A|X}P_{Y|A,X}$, where the source distributions $P_{X}$ and $P_{Y|A,X}$ are given.

\underline{Code construction:}
\begin{itemize}
\item Generate a codebook $\mathcal{C}$, consisting of $2^{n(I(Y;A)-\epsilon)}$ independent codewords, $\vA(i)$, $i\in\{1,\dots,2^{n(I(Y;A)-\epsilon)}\}$, where each element is i.i.d. $\sim\prod_{j=1}^{n} P_{A}(a_j)$, for some $\epsilon > 0$. Randomly bin the codewords in $\mathcal{C}$ into $2^{n\Delta}$ bins, where $\Delta=(I(Y;A) - I(X;A) - 2\epsilon)$, such that in each bin there are $2^{n(I(X;A)+\epsilon)}$ codewords. For each $\vA\in\mathcal{C}$, the bin that contains $\vA$ will be denoted as $\mathcal{B}_{\vA}$. Each bin can be represented by a message of $n\Delta$ bits, which is the rate that is reduced from the source $\vX$. Let $\IA$ denote the representation of each codeword by $\ceil{I(Y;A)-\epsilon}$ elements from $\mathbb{F}_{2^{n}}$.
\item The $\vX$ sequences are randomly binned into $2^{nr_1}$ bins, where $r_1 \triangleq H(X|Y,A)+\epsilon$. The notation $\mathcal{B}_{\vX}$ stands for the first $n\Delta$ bits of the bin index where $\vX$ falls. Additionally, each bin index is denoted by $\IX(j)$, $j \in \{1,\dots,2^{nr_1}\}$, consisting of $\ceil{r_1}$ elements from the finite field $\mathbb{F}_{2^{n}}$.
\item  The $\vY$ sequences are randomly binned into $2^{nr_2}$ bins, where $r_2 \triangleq H(Y)+\epsilon$. Each bin is represented by a vector consisting of $\ceil{r_2}$ elements from $\mathbb{F}_{2^{n}}$, and denoted by $\IY(k)$ $k \in \{1,\dots,2^{nr_2}\}$.
\item The process of network coefficients generation is the same as for the case $I(X;A)-I(Y;A)\geq 0$, and therefore omitted here.
\end{itemize}

\underline{Encoding:}
Given the source realization $\vx$, node $s_1$ looks in the actions' bin satisfying $\mathcal{B}_{\vA}=\mathcal{B}_{\vx}$ for a codeword $\vA$ which is jointly typical with $\vx$. The source $\vy$ is then generated and available at node $s_2$. The input to the network will then be the vector $[\Ix, \Ia, \Iy ]$ corresponding to the bins where the source sequences fall and the chosen actions codeword.

\underline{Decoding:}
Having received the vector $\underline{Z}_{t}$, each node $t\in\tau$ looks for a unique triplet $(\vX,\vA,\vY)\in \strtyp(X,A,Y)$ satisfying $[\vX, \vA, \vY]\net=\underline{Z}_{t}$ and $\mathcal{B}_{\vA}=\mathbf{I_{\vX}}$.

\subsection{An Upper Bound in Randomized Networks}\label{subsec:lemma}
Following the result in \cite[Appendix \rom{1}]{random}, the next lemma provides an upper bound on the probability of the event that two different inputs to a randomized linear network yield the same output at a terminal node $t$. Due to the fact that the network is linear, this event is equivalent to the event that the difference between two inputs yields the zero processes at the terminal node. The next lemma will be at the assist of our direct proof. Moreover, as we will see in Example. \ref{lemma:example}, it has implications beyond the scoop of our proof as well.

Let $\mathcal{G}=(\mathcal{V},\mathcal{E})$ be a directed, acyclic graph. The matrix $\net$ represents the complete mapping of the network from inputs to some node $t$, where each non-zero element in this matrix is generated uniformly from $\mathbb{F}_{2^{n}}$. Now, consider a set of sources with no incoming links, denoted by $\mathcal{S}\subseteq\mathcal{V}$, such that $\mathcal{S}=\{1,\dots,k\}$. Each node $i\in\mathcal{S}$ consists of a vector, $\underline{u}_i$, which comprises elements from $\mathbb{F}_{2^{n}}$. For any two different inputs to the network, denoted by $\underline{u}=[\underline{u}_1 \underline{u}_2 \dots \underline{u}_k]$ and $\underline{v}=[\underline{v}_1 \underline{v}_2 \dots \underline{v}_k]$, let $\mathcal{W}$ be a subset of $\mathcal{S}$, such that if $\underline{u}_i \neq \underline{v}_i$ then $i\in\mathcal{W}$.
\begin{lemma}\label{lemma:bounds}
For any pair of different inputs $\underline{u}$ and $\underline{v}$, the probability that these inputs induce the same output in node $t$ is bounded by:
\begin{align}\label{eq:lemma}
\Pr\left([\underline{u}-\underline{v}]\net=\mathbf{0}\right)&\leq \left(\frac{L}{2^n}\right)^{\mathcal{C}(\mathcal{V}^{\ast}_{\mathcal{W};t})},
\end{align}
where $L$ denotes the maximum source-receiver path length, and $\mathcal{C}(\mathcal{V}^{\ast}_{\mathcal{W};t})$ is the minimum cut-set between $\mathcal{W}$ and $t$.
\end{lemma}
Note that the upper bound is independent of the number of elements in the vector $\underline{u}_i, \forall i$. This remarkable fact allows us to think of $\IA$ and $\IX$ as the same input in our network; thus, we have the same upper bound on two different probabilities in our analysis:
\begin{align}
\Pr\left([\Itx-\Ix, \Ita-\Ia, \mathbf{0}]\net=\mathbf{0}| \Itx\neq\Ix, \Ita\neq\Ia\right)&\leq  \left(\frac{L_1}{2^n}\right)^{\mathcal{C}(\mathcal{V}^{\ast}_{s_1;t})},\\
\Pr\left([\Itx-\Ix, \mathbf{0}, \mathbf{0}]\net=\mathbf{0}| \Itx\neq\Ix\right) &\leq  \left(\frac{L_1}{2^n}\right)^{\mathcal{C}(\mathcal{V}^{\ast}_{s_1;t})},
\end{align}
where $L_1$ is the maximum length of a path between $s_1$ and $t$.

We now show how the lemma above can serve as an easy and elegant proof to the capacity of multicast networks. A sender wishes to transmit a message to a set of terminal nodes through a directed, acyclic network. The sender transmits a message from the set $\mathcal{M}=\{1,\dots,2^{nR}\}$, and each receiver $t\in\tau$ is required to decode the correct message in a lossless manner. We want to characterize the single-letter expression for the maximal rate $R$ that can be used for a reliable communication in a given network.
\begin{example}[Multicast network]\label{lemma:example}
Consider a directed, acyclic network, where sender denoted as node $1$ is required to transmit a message from $\mathcal{M}=\{1,\dots,2^{nR}\}$ to a set of terminal nodes denoted as $\tau$. The sender can choose any message, $m\in\mathcal{M}$, and each receiver $t\in\tau$ is required to decode the correct message in a lossless manner. We provide here a simple $n$ block-length coding scheme follows by an analysis of the probability of error.

To encode the message, we rely on the scalar algebraic approach we have shown earlier in the code construction of the proof for Theorem \ref{theorem:network}. The input to the network is $\underline{m}$, where $\underline{m}$ is a vector representing $m$ by elements from $\mathbb{F}_{2^{n}}$. Each terminal node, $t\in\tau$, having received $\underline{z}_t$ looks for $m\in\mathcal{M}$ satisfying $\underline{m}\net=\underline{z}_t$.

Now, assume without loss of generality that the message $m$ was sent. An error occurs only if there exists $m'\neq m$ satisfying $\underline{m}'\net=\underline{z}_t$ for some $t\in\tau$.

Upper bounding the probability of error for some receiver $t\in\tau$ yields:
\begin{align}\label{lemma:example_eq}
    \Pr(\text{error})&= \Pr(\exists \tilde{m}\neq m:[\underline{\tilde{m}}- \underline{m}]\net=\mathbf{0})\nonumber\\
    &=\sum_{\tilde{m}\in\mathcal{M}}\Pr([\underline{\tilde{m}}-\underline{m}]\net=\mathbf{0})\nonumber\\
    &\leq 2^{nR}\left(\frac{L}{2^n}\right)^{\mathcal{C}(\mathcal{V}^{\ast}_{1;t})}\nonumber\\
    &= L^{\mathcal{C}(\mathcal{V}^{\ast}_{1;t})} 2^{n(R-\mathcal{C}(\mathcal{V}^{\ast}_{1;t}))}.
\end{align}
Note that if $R\leq\mathcal{C}(\mathcal{V}^{\ast}_{1;t})$, the term \eqref{lemma:example_eq} tends to zero for sufficiently large $n$. Our requirement is to decode the message correctly at all the receivers; thus, using the union-bound we achieve that the overall probability tends to zero for large $n$ if,
\begin{equation}
    R < \min_{t\in\tau}\mathcal{C}(\mathcal{V}^{\ast}_{1;t}),
\end{equation}
which is the known multicast result.
\end{example}
\begin{proof}[Proof of Lemma \ref{lemma:bounds}]
Let $\mathcal{G}_1$ be a subgraph of $\mathcal{G}$ consisting of all links downstream of $\mathcal{W}$, where a link $l$ is considered downstream if $o(l)\in \mathcal{W}$, or if there is a directed path from some source $s\in \mathcal{W}$ to $o(l)$. Since information sources can differ only in source nodes satisfying $i\in\mathcal{W}$, this fact induces that only links in $\mathcal{G}_1$ will affect the bound on probability.

Note that in a random linear network code, any link $l$ which has at least one nonzero input transmits the zero process with probability $2^{-nc_l}$, where $c_l$ is the capacity of $l$. This is the same as the probability that a pair of distinct values for the inputs of $l$ are mapped to the same output value on $l$.

For a given pair of distinct input values, let $E_l$ be the event where the corresponding inputs to link $l$ are distinct, but the corresponding values on $l$ are the same. Let $E(\mathcal{G}_1)$ be the event that $E_l$ occurs for some link $l$ on every source-terminal path in graph $\mathcal{G}_1$. Note, the probability of the event $E(\mathcal{G}_1)$ is equal to the probability that two inputs induce the same output at the terminal node, i.e. $\Pr([\underline{u}-\underline{v}]\net=\mathbf{0})$.

We proceed and look at the set of source-terminal paths in the graph $\mathcal{G}_1$. Since there exists $C(\mathcal{V}^{\ast}_{\mathcal{W};t})$ disjoint paths, we denote each disjoint path as $\mathcal{P}_{\mathcal{G}_1i}$ with its corresponding length $L_i$, where $i\in\{1,\dots,C(\mathcal{V}^{\ast}_{\mathcal{W};t})\}$. Furthermore, we denote $E(\mathcal{P}_{\mathcal{G}_1i})$ as the event that $E_l$ occurs for some link on $\mathcal{P}_{\mathcal{G}_1i}$.
\begin{align}\label{eq:lemmaproof}
    \Pr(E(\mathcal{G}_1))&= \Pr \left(\bigcap_{i=1}^{C(\mathcal{V}^{\ast}_{\mathcal{W};t})} E(\mathcal{P}_{\mathcal{G}_1i})\right)\nonumber\\
                         &\stackrel{(a)}= \prod_{i=1}^{C(\mathcal{V}^{\ast}_{\mathcal{W};t})} \Pr(E(\mathcal{P}_{\mathcal{G}_1i}))\nonumber\\
                         &=\prod_{i=1}^{C(\mathcal{V}^{\ast}_{\mathcal{W};t})} 1-\left(1-\frac{1}{2^n}\right)^{L_i}\nonumber\\
                         &\stackrel{(b)}\leq \prod_{i=1}^{C(\mathcal{V}^{\ast}_{\mathcal{W};t})} 1-\left(1-\frac{1}{2^n}\right)^{L}\nonumber\\
                         &=  \left(1-\left(1-\frac{1}{2^n}\right)^{L}\right)^{C(\mathcal{V}^{\ast}_{\mathcal{W};t})}\nonumber\\
                         &\stackrel{(c)}\leq \left(\frac{L}{2^n}\right) ^ {C(\mathcal{V}^{\ast}_{\mathcal{W};t})},
\end{align}
where:
\begin{itemize}
  \item [(a)] follows from the fact that the coefficients are generated independently on each path;
  \item [(b)] follows from $L=\max_{i} L_i$;
  \item [(c)] follows from applying Bernoulli's inequality, i.e. $(1 + x)^{r} \geq 1 + rx$, with substituting $x = -\frac{1}{2^n}$ and  $r = L$.
\end{itemize}
\end{proof}

\subsection{Generalized Cut-Set Bounds (Converse)}\label{subsec:converse}
In this subsection, we derive an outer bound on the set of achievable rates for our model. The outer bound is indeed a generalization of the known cut-set bound, this method of generalized cut-set bounds was adopted also in \cite{Asaf_networks_SI}.
%Here, we show that the outer bound is indeed tight and coincides with \eqref{NC:theoremeq}.

For the converse of Theorem \ref{theorem:network}, given an achievable $\left(\left(2^{nR_{l}}\right)_{l\in\mathcal{E}},n\right)$ source code we need to show that there exists a joint distribution, $P_{X,A,Y}=P_{X}P_{A|X}P_{Y|X,A}$, such that the inequalities in Theorem \ref{theorem:network} hold.

For any set of messages denoted by $\mathcal{M}_1$, across a cut $\mathcal{V}_{s_1;t}$, we have
\begin{align}
n \mathcal{C}(\mathcal{V}_{s_1;t})&\geq H(\mathcal{M}_1)\nonumber\\
    &= H(\mathcal{M}_1) + H(X^{n}|Y^{n},\mathcal{M}_1) - H(X^{n}|Y^{n},\mathcal{M}_1)\nonumber\\
    &\stackrel{(a)}\geq H(\mathcal{M}_1) + H(X^{n}|Y^{n},\mathcal{M}_1) - n\epsilon_n\nonumber\\
    &\stackrel{(b)}\geq I(\mathcal{M}_1;X^{n},Y^{n}) + H(X^{n}|Y^{n},\mathcal{M}_1) - n\epsilon_n\nonumber\\
    &= H(X^{n},Y^{n}) - H(X^{n},Y^{n}|\mathcal{M}_1) + H(X^{n}|Y^{n},\mathcal{M}_1) - n\epsilon_n\nonumber\\
    &= H(X^{n}) + H(Y^{n}|X^{n}) - H(Y^{n}|\mathcal{M}_1) - n\epsilon_n\nonumber\\
    &\stackrel{(c)}\geq H(X^{n}) + H(Y^{n}|X^{n},A^{n}) - H(Y^{n}) - n\epsilon_n\nonumber\\
    &\stackrel{(d)}\geq \sum_{i=1}^{n} \left[ H(X_{i})-H(Y_{i}) + H(Y_{i}|A_{i},X_{i})\right]- n\epsilon_n\nonumber\\
    &= \sum_{i=1}^{n} \left[H(X_{i}) - H(X_{i}|A_{i})+ H(X_{i}|A_{i}) + H(Y_{i}|A_{i},X_{i}) - H(Y_{i})\right]- n\epsilon_n\nonumber\\
    &= \sum_{i=1}^{n} \left[I(X_{i};A_{i}) + H(Y_{i}|A_{i}) + H(X_{i}|A_{i},Y_{i}) - H(Y_{i})\right]- n\epsilon_n\nonumber\\
    &= \sum_{i=1}^{n} \left[I(X_{i};A_{i}) - I(Y_{i};A_{i}) + H(X_{i}|A_{i},Y_{i})\right] - n\epsilon_n,
\end{align}
where:
\begin{itemize}
  \item [(a)] follows from Fano's inequality;
  \item [(b)] follows from the fact that $\mathcal{M}_1$ is a deterministic function of $X^{n},Y^{n}$;
  \item [(c)] follows from the fact that $A^{n}$ is a deterministic function of $X^{n}$;
  \item [(d)] follows from the $X^{n}$ is memoryless, conditioning reduces entropy and the memoryless property \eqref{setup:memoryless}.
\end{itemize}
For the second inequality in \eqref{NC:theoremeq}, we have
\begin{align}
n\mathcal{C}(\mathcal{V}_{s_2;t})&\geq H(\mathcal{M}_2)\nonumber\\
    &\geq H(\mathcal{M}_2,Y^{n}|X^{n}) - H(Y^{n}|X^{n},\mathcal{M}_2)\nonumber\\
    &\stackrel{(a)}\geq H(\mathcal{M}_2,Y^{n}|X^{n}) - n\epsilon_n\nonumber\\
    &\stackrel{(b)}= H(Y^{n}|X^{n},A^{n}) + H(\mathcal{M}_2|X^{n},Y^{n}) - n\epsilon_n\nonumber\\
    &\stackrel{(c)}= H(Y^{n}|X^{n},A^{n}) - n\epsilon_n\nonumber\\
    &= \sum_{i=1}^{n} H(Y_{i}|A_{i},X_{i}) - n\epsilon_n,
\end{align}
where:
\begin{itemize}
  \item [(a)] follows from Fano's inequality;
  \item [(b)] follows from the fact that $A^{n}$ is a deterministic function of $X^n$;
  \item [(c)] follows from the fact that $\mathcal{M}_2$ is a deterministic function of $X^n,Y^n$.
\end{itemize}
For the sum-rate, we have
\begin{align}
n\mathcal{C}(\mathcal{V}_{s_1,s_2;t})&\geq H(\mathcal{M}_3)\nonumber\\
    &= H(X^{n},Y^{n},\mathcal{M}_3) - H(X^{n},Y^{n}|\mathcal{M}_3)\nonumber\\
    &\stackrel{(a)}\geq H(X^{n},Y^{n},\mathcal{M}_3) - n\epsilon_n\nonumber\\
    &\stackrel{(b)}= H(X^{n},Y^{n})- n\epsilon_n\nonumber\\
    &\stackrel{(c)}= H(X^{n}) + H(Y^{n}|X^{n},A^{n}) - n\epsilon_n\nonumber\\
    &\stackrel{(d)}= \sum_{i=1}^{n} \left[H(X_{i}) + H(Y_{i}|X_{i},A_{i})\right] - n\epsilon_n\nonumber\\
    &= \sum_{i=1}^{n} \left[H(X_{i}) - H(X_{i}|A_{i}) + H(X_{i}|A_{i}) + H(Y_{i}|X_{i},A_{i})\right] - n\epsilon_n\nonumber\\
    &= \sum_{i=1}^{n} \left[I(X_{i};A_{i}) + H(Y_{i},X_{i}|A_{i})\right] - n\epsilon_n,
\end{align}
where:
\begin{itemize}
  \item [(a)] follows from Fano's inequality;
  \item [(b)] follows from the fact that $\mathcal{M}_3$ is a deterministic function of $X^n,Y^n$;
  \item [(c)] follows from the fact that $A^n$ is a deterministic function of $X^n$;
  \item [(d)] follows from the fact the $X^{n}$ is memoryless and the memoryless property \eqref{setup:memoryless}.
\end{itemize}
Let us summarize the lower bounds we have characterized:
\begin{align}\label{converse}
    & \mathcal{C}(\mathcal{V}_{s_1;t})\geq \sum_{i=1}^{n} \frac{1}{n}\left[I(X_{i};A_{i}) - I(Y_{i};A_{i}) + H(X_{i}|A_{i},Y_{i})\right] - \epsilon_{n},\nonumber \\
    & \mathcal{C}(\mathcal{V}_{s_2;t})\geq  \sum_{i=1}^{n} \frac{1}{n}H(Y_{i}|A_{i},X_{i}) - \epsilon_n, \nonumber\\
    & \mathcal{C}(\mathcal{V}_{s_1,s_2;t})\geq \sum_{i=1}^{n} \frac{1}{n}\left[I(X_{i};A_{i}) + H(Y_{i},X_{i}|A_{i})\right] - \epsilon_n,
\end{align}
for some cuts $\mathcal{V}_{s_1;t},\mathcal{V}_{s_2;t},\mathcal{V}_{s_1,s_2;t}$.

To complete the proof, we minimize the left hand side of \eqref{converse} by taking the cuts to be $\mathcal{C}(\mathcal{V}^{\ast}_{s_{1};t}), \mathcal{C}(\mathcal{V}^{\ast}_{s_{2};t}),$ and $\mathcal{C}(\mathcal{V}^{\ast}_{s_{1},s_{2};t})$, respectively. Derivation of the single-letter characterization in \eqref{converse} is done by common time-sharing technique.
\section{Conclusions And Future Work}\label{section:conclusions}
In the current work, we have considered the setup of correlated sources with action-dependent joint distribution. Specifically, the optimal rate regions were characterized for the case where actions taken at the decoder and for the case of actions taken at the encoder. Further, we have presented the set of achievable rates for a scenario where action-dependent sources are known at different nodes of a general network and are required at a set of terminal nodes. Remarkably, RLNC was proved to be optimal also for this scenario, even though this is not a multicast problem. Moreover, the set of achievable rates involved mutual information terms, which are not typical in multicast problems. Two binary examples were studied, and it was shown how actions affect the achievable rate region in a non-trivial manner.

As can be seen from this and additional work \cite{PermuterWiessman_vending_11,zhao14_compression_actions,in_block,action_double_sided,information_embedding,ChiaAsnaniWeissman13_multi_terminal_source_coding_action}, actions have a significant impact on the set of achievable rates in source coding problems and many classical source coding problems can be extended using actions. One particular, as yet unsolved, source coding problem that would be interesting to study is the case of action-dependent source coding with a helper. In this scenario the considered setup is of correlated sources with actions, yet only a reconstruction of $\vX$ is required at the decoder. In the source coding helper problem, the sequence $\vY$ which is being transmitted on a rate-limited link plays the role of SI and not of an information source as in our model. The main difficulty in proving the converse follows from the fact that $\vY$ is not distributed i.i.d. as in the original problem of source coding with a helper \cite{Ahlswede-Korner75}.

\appendix[Analysis of the probability of error for the direct of Theorem \ref{theorem:network}]\label{app:analysis}
Following the direct method in Section \ref{section_network}, the probability of error is analyzed for both cases: a negative and positive sign of the term $I(X;A)-I(Y;A)$.
\subsection{For the case $I(X;A)-I(Y;A)\geq0$}
The events corresponding to possible encoding and decoding errors are as follows:
An encoding error occurs if:
\begin{equation}
  \mathcal{E}_{1} = \{\not\exists i: (\vx,\vA(i))\in \strtyp(X,A)\}.
\end{equation}
For the events of decoding errors, we derive upper bounds for some terminal node $t\in\tau$. Later on, we conclude the complete achievable region by a union bound on all $t\in\tau$. For a terminal node $t\in\tau$, a decoding error will occur for any of the next events:
\begin{align}
  \mathcal{E}_{2} &= \{(\vX,\vA,\vY)\not\in\ \strtyp(X,A,Y)\}, \\
  \mathcal{E}_{3} &= \{\exists \vtX\neq \vX,\vtA\neq \vA: [\ItX, \ItA, \IY]\netz, (\vtX,\vtA,\vY)\in\ \strtyp(X,A,Y)\},\\
  \mathcal{E}_{4} &= \{\exists \vtX\neq \vX:
  [\ItX, \IA, \IY]\netz, (\vtX,\vA,\vY)\in\ \strtyp(X,A,Y)\},\\
  \mathcal{E}_{5} &= \{\exists \vtY\neq \vY:
  [\IX, \IA, \ItY]\netz, (\vX,\vA,\vtY)\in\ \strtyp(X,A,Y)\},\\
  \mathcal{E}_{6} &= \{\exists \vtX\neq \vX,\vtY\neq \vY:
  [\ItX, \IA, \ItY]\netz, (\vtX,\vA,\vtY)\in\ \strtyp(X,A,Y)\},\\
  \mathcal{E}_{7} &= \{\exists \vtX\neq \vX,\vtA\neq \vA,\vtY\neq \vY:
  [\ItX, \ItA, \ItY]\netz, (\vtX,\vtA,\vtY)\in\ \strtyp(X,A,Y)\}.
\end{align}
The total probability of an error can be bounded as:
\begin{align}
    P_{e}^{(n)} &= \Pr(\bigcup_{i=1}^{7}\mathcal{E}_{i}) \nonumber\\
                &\leq \Pr(\mathcal{E}_1\bigcup \mathcal{E}_{2}) + \sum_{i=3}^{7}\Pr(\mathcal{E}_i)\nonumber\\
                &\leq \Pr(\mathcal{E}_1) + \Pr(\mathcal{E}_2|\mathcal{E}_{1}^{C}) + \sum_{i=3}^{7}\Pr(\mathcal{E}_i).
\end{align}
Therefore, we can upper bound each term separately.
%Note that the binning rates $r_1,r_2$ are greater than the entropy $H(X),H(Y)$ of each source, respectively. According to the source-coding theorem \cite[Theorem 3.4]{ElGamal}, the probability that a bin contains two typical sequences tends to zero as $n\rightarrow\infty$. Hence, throughout the analysis we can assume that if $\vX\neq\vtX$ are typical sequences, then $\IX\neq\ItX$. The same argument stands for the $\vY$ sequences, i.e. $\vY\neq\vtY$ implies $\IY\neq\ItY$.
\begin{itemize}
  \item[1.] For $\mathcal{E}_{1}$, it is known from the covering lemma \cite[Lemma 3.3]{ElGamal} that $\Pr(\mathcal{E}_1)\rightarrow0$ for $n\rightarrow\infty$ if we fix  $r_A= I(X;A)+\epsilon$.
  \item[2.] Given the event $\mathcal{E}_1^{C}$, and the fact that $Y^{n}$ is generated as the output of a memoryless channel, we use the conditional typicality lemma \cite[Chapter 2]{ElGamal} to show that $\Pr(\mathcal{E}_2|\mathcal{E}_{1}^{C})\rightarrow0$ as $n\rightarrow\infty$.
  \item[3.] To upper-bound $\mathcal{E}_{3}$, we have
  \begin{align}
    &\Pr(\mathcal{E}_{3}) \nonumber\\
    &= \Pr \left( \exists \vtX \neq \vX,\vtA \neq \vA:[\ItX-\IX, \ItA-\IA, \mathbf{0}]\net =\mathbf{0}, (\vtX,\vtA,\vY)\in \strtyp\right)\nonumber\\
  &= \sum\limits_{(\vx,\va,\vy)}\mspace{-25mu}P(\vx,\va,\vy) \Pr\left(\exists\vtX\mspace{-8mu}\neq\vx,\vtA\neq \va:
  [\ItX\mspace{-8mu}-\Ix, \ItA\mspace{-4mu}-\mspace{-4mu}\Ia, \mathbf{0}]\net=\mathbf{0}, (\vtX,\vtA,\vy)\in \mspace{-4mu}\strtyp\right)\nonumber\\
  &= \mspace{-8mu}\sum\limits_{(\vx,\va,\vy)}\mspace{-25mu} P(\vx\mspace{-4mu},\va\mspace{-8mu},\vy\mspace{-2mu})\sum_{\vta\in\mathcal{Q}}\mspace{-10mu}\Pr\left(\exists \vtX\neq \vx \mspace{-8mu}:[\ItX\mspace{-8mu}-\Ix, \Ita\mspace{-8mu}-\Ia, \mathbf{0}]\net=\mathbf{0},(\vtX,\vta,\vy)\in \mspace{-4mu}\strtyp|(\vta\mspace{-8mu},\vy)\in\mspace{-8mu} \strtyp\right),\nonumber\\
  & \text{where $\mathcal{Q}:=\{\vta\in \mathcal{C}:\vta\neq\va,(\vta,\vy)\in \strtyp(A|Y)\}$} \nonumber\\
%  &= \sum\limits_{(\vx,\va,\vy)} \mspace{-19mu}P(\vx,\va,\vy)\sum_{\vta\in\mathcal{Q}}\sum\limits_{\substack{\vtx\neq\vx:\\ \vtx\in \strtyp(X|Y,A)}}\mspace{-30mu}\Pr\left([\Itx-\Ix, \Ita-\Ia, \mathbf{0}]\net=\mathbf{0}| (\vtx,\vta,\vy)\in \strtyp\right)\\
  &= \sum\limits_{(\vx,\va,\vy)} P(\vx,\va,\vy) \sum_{\vta\in\mathcal{Q}}\sum\limits_{\substack{\vtx\neq\vx:\\ \vtx\in \strtyp(X|Y,A)}}\Pr\left([\Itx-\Ix, \Ita-\Ia, \mathbf{0}]\net=\mathbf{0}\right)\nonumber\\
  &\stackrel{(a)}\leq \sum\limits_{(\vx,\va,\vy)} P(\vx,\va,\vy)\sum_{\vta\in\mathcal{Q}}\sum\limits_{\substack{\vtx\neq\vx:\\ \vtx\in \strtyp(X|Y,A)}}\left(\frac{L_1}{2^n}\right)^{\mathcal{C}(\mathcal{V}^{\ast}_{s_{1};t})}\nonumber\\
  &\stackrel{(b)}\leq \sum\limits_{(\vx,\va,\vy)} P(\vx,\va,\vy)2^{n(r_A-I(Y;A)+2\epsilon)}|\strtyp(X|Y,A)|\left(\frac{L_1}{2^n}\right)^{\mathcal{C}(\mathcal{V}^{\ast}_{s_{1};t})}\nonumber\\
  &\leq \sum\limits_{(\vx,\va,\vy)} P(\vx,\va,\vy)2^{n(I(X;A)-I(Y;A)+H(X|Y,A)+3\epsilon)}\left(\frac{L_1}{2^n}\right)^{\mathcal{C}(\mathcal{V}^{\ast}_{s_{1};t})}\nonumber\\
  &\leq  2^{n(I(X;A)-I(Y;A)+H(X|Y,A)+3\epsilon)}\left(\frac{L_1}{2^{n}}\right)^{\mathcal{C}(\mathcal{V}^{\ast}_{s_{1};t})},
 \end{align}
where:
\begin{itemize}
  \item [(a)] follows from applying Lemma \ref{lemma:bounds}. The notation $L_1$ denotes the maximum path length between $s_1$ and $t$. Note that the binning rate $r_1$ is greater than the source entropy $H(X)$. According to the source-coding theorem \cite[Theorem 3.4]{ElGamal}, the probability that a bin contains two typical sequences tends to zero as $n\rightarrow\infty$. Hence, we can assume that if $\vX\neq\vtX$ are two typical sequences, then $\IX\neq\ItX$;
  \item [(b)] follows from deriving an upper bound on $|\mathcal{Q}|$. Namely, we are interested in the amount of codewords in $\mathcal{C}$ that are jointly typical with $\vy$. One may think of it as a random binning of the codebook at a rate of $r_A-I(Y;A)-2\epsilon$, such that in each bin there are $I(Y;A)-\epsilon$ sequences. Since $\vy$ was generated according to $\va$, which is different from $\vta$, then with high probability there will be only one sequence in each bin that is jointly typical with $\vy$. Therefore, the amount of $\vta$ satisfying $\vta\in\mathcal{Q}$ is bounded by the number of bins, e.g. $2^{n(r_A-I(Y;A)-2\epsilon)}$.
\end{itemize}
  \item[4.] To upper-bound $\Pr(\mathcal{E}_{4})$, we have
  \begin{align}
  \Pr(\mathcal{E}_{4}) &= \Pr \left( \exists \vtX\neq \vX:
  [\ItX-\IX, \mathbf{0}, \mathbf{0}]\net =\mathbf{0}, (\vtX,\vA,\vY)\in \strtyp\right)\nonumber\\
  &= \sum\limits_{(\vx,\va,\vy)} P(\vx,\va,\vy)\Pr\left( \exists \vtX\neq\vx:
  [\ItX-\Ix, \mathbf{0}, \mathbf{0}]\net=\mathbf{0}, (\vtX,\va,\vy)\in \strtyp\right)\nonumber\\
  &= \sum_{(\vx,\va,\vy)} P(\vx,\va,\vy)\sum\limits_{\substack{\vtx\neq\vx:\\\vtx\in \strtyp(X|Y,A)}}\mspace{-30mu}\Pr\left([\Itx-\Ix, \mathbf{0}, \mathbf{0}]\net=\mathbf{0}| (\vtx,\va,\vy)\in \strtyp\right)\nonumber\\
  &\leq \sum\limits_{(\vx,\va,\vy)} P(\vx,\va,\vy)|\strtyp(X|Y,A)|\left(\frac{L_1}{2^n}\right)^{\mathcal{C}(\mathcal{V}^{\ast}_{s_{1};t})}\nonumber\\
  &\leq \sum\limits_{(\vx,\va,\vy)} P(\vx,\va,\vy)2^{n(H(X|Y,A)+\epsilon)}\left(\frac{L_1}{2^n}\right)^{\mathcal{C}(\mathcal{V}^{\ast}_{s_{1};t})}\nonumber\\
  &\leq  2^{n(H(X|Y,A)+\epsilon)}\left(\frac{L_1}{2^n}\right)^{\mathcal{C}(\mathcal{V}^{\ast}_{s_{1};t})}.
  \end{align}
  \item[5.] To upper-bound $\Pr(\mathcal{E}_{5})$, we have
  \begin{align}
  \Pr(\mathcal{E}_{5}) &= \Pr \left( \exists \vtY\neq \vY:
  [\mathbf{0}, \mathbf{0}, \ItY-\IY]\net=\mathbf{0}, (\vX,\vA,\vtY)\in \strtyp\right)\nonumber\\
  &= \sum\limits_{(\vx,\va,\vy)} P(\vx,\va,\vy)\Pr\left(\exists \vtY\neq \vy:
  [\mathbf{0}, \mathbf{0}, \ItY-\Iy]\net=\mathbf{0}, (\vx,\va,\vtY)\in \strtyp\right)\nonumber\\
  &= \sum\limits_{(\vx,\va,\vy)} P(\vx,\va,\vy)\sum\limits_{\substack{\vty\neq\vy:\\\vty\in \strtyp(Y|X,A)}} \Pr\left([\mathbf{0}, \mathbf{0}, \Ity-\Iy]\net=\mathbf{0}\right)\nonumber\\
  &\leq \sum\limits_{(\vx,\va,\vy)} P(\vx,\va,\vy)|\strtyp(Y|X,A)|\left(\frac{L_2}{2^n}\right)^{\mathcal{C}(\mathcal{V}^{\ast}_{s_{2};t})}\nonumber\\
  &= \sum\limits_{(\vx,\va,\vy)} P(\vx,\va,\vy)2^{n(H(Y|X,A)+\epsilon)}\left(\frac{L_2}{2^n}\right)^{\mathcal{C}(\mathcal{V}^{\ast}_{s_{2};t})}\nonumber\\
  &\leq  2^{n(H(Y|X,A)+\epsilon)}\left(\frac{L_2}{2^n}\right)^{\mathcal{C}(\mathcal{V}^{\ast}_{s_{2};t})}.
  \end{align}
  \item[6.] To upper-bound $\Pr(\mathcal{E}_{6})$, we have
  \begin{align}
  &\Pr(\mathcal{E}_{6}) \nonumber\\
  &= \Pr \left( \exists \vtX\neq \vX,\vtY\neq \vY: [\ItX-\IX, \mathbf{0}, \ItY-\IY]\net =\mathbf{0}, (\vtX,\vA,\vtY)\in \strtyp\right)\nonumber\\
  &= \mspace{-12mu}\sum\limits_{(\vx,\va,\vy)} \mspace{-15mu}P(\vx,\va,\vy)\Pr\left( \exists \vtX\neq\vx,\vtY\neq \vy:
      [\ItX-\Ix, \mathbf{0}, \ItY-\Iy]\net =\mathbf{0}, (\vtX,\va,\vtY)\in \strtyp\right)\nonumber\\
  &= \sum_{(\vx,\va,\vy)} P(\vx,\va,\vy)\sum\limits_{\substack{\vtx\neq\vx,\vty\neq\vy:\\(\vtx,\va,\vty)\in \strtyp(X,Y|A)}}\Pr\left([\Itx-\Ix, \mathbf{0}, \Ity-\Iy]\net=\mathbf{0}\right)\nonumber\\
  &\leq \sum\limits_{(\vx,\va,\vy)} P(\vx,\va,\vy)|\strtyp(X,Y|A)|\left(\frac{L_3}{2^n}\right)^{\mathcal{C}(\mathcal{V}^{\ast}_{s_{1},s_{2};t})}\nonumber\\
  &\leq  2^{n(H(X,Y|A)+\epsilon)}\left(\frac{L_3}{2^n}\right)^{\mathcal{C}(\mathcal{V}^{\ast}_{s_{1},s_{2};t})}.
  \end{align}
  \item[7.] To upper-bound $\Pr(\mathcal{E}_{7})$, we have
  \begin{align}
  &\Pr(\mathcal{E}_{7})\nonumber\\
  &= \Pr \left( \exists \vtX\neq \vX,\vtA\neq \vA,\vtY\neq \vY: [\ItX-\IX, \ItA-\IA, \ItY-\IY]\net =\mathbf{0}, (\vtX,\vtA,\vtY)\in \strtyp\right)\nonumber\\
%  &= \mspace{-20mu}\sum\limits_{(\vx,\va,\vy)} \mspace{-20mu}P(\vx,\va,\vy)\Pr\left( \exists \vtX\mspace{-8mu}\neq\vx, \vtA\mspace{-8mu}\neq\va, \vtY\mspace{-8mu}\neq\vy:
%  [\ItX\mspace{-8mu}-\Ix, \ItA\mspace{-8mu}-\Ia, \ItY\mspace{-8mu}-\Iy]\net =\mathbf{0}, (\vtX,\vtA,\vtY)\mspace{-6mu}\in\mspace{-6mu} \strtyp\right)\nonumber\\
  &= \mspace{-20mu}\sum\limits_{(\vx,\va,\vy)} \mspace{-20mu}P(\vx,\va,\vy)\cdot\nonumber \\
  & \sum_{\vta\neq\va:\vta\in\mathcal{C}}\mspace{-30mu}\Pr\left(\exists \vtX\neq\vx,\vtY\neq\vy:
  [\ItX\mspace{-6mu}-\mspace{-6mu}\Ix, \Ita\mspace{-6mu}-\mspace{-6mu}\Ia, \ItY\mspace{-6mu}-\mspace{-6mu}\Iy]\net =\mathbf{0}, (\vtX,\vta,\vtY)\mspace{-6mu}\in \mspace{-6mu}\strtyp\right)\nonumber\\
  &= \mspace{-20mu}\sum\limits_{(\vx,\va,\vy)} \mspace{-20mu}P(\vx,\va,\vy)\sum_{\vta\neq\va:\vta\in\mathcal{C}}\mspace{-20mu}\sum\limits_{\substack{\vtx\neq\vx,\vty\neq\vy:\\(\vtx,\vta,\vty)\in \strtyp(X,Y|A)}}\mspace{-20mu}\Pr\left([\Itx-\Ix, \Ita-\Ia, \Ity-\Iy]\net =\mathbf{0}\right)\nonumber\\
  &\leq \sum\limits_{(\vx,\va,\vy)} P(\vx,\va,\vy)|\mathcal{C}||\strtyp(X,Y|A)|\left(\frac{L_3}{2^n}\right)^{\mathcal{C}(\mathcal{V}^{\ast}_{s_{1},s_{2};t})}\nonumber\\
  &\leq \sum\limits_{(\vx,\va,\vy)} P(\vx,\va,\vy)2^{n(I(X;A)+\epsilon)}2^{n(H(X,Y|A)+\epsilon)}\left(\frac{L_3}{2^n}\right)^{\mathcal{C}(\mathcal{V}^{\ast}_{s_{1},s_{2};t})}\nonumber\\
  &\leq  2^{n(I(X;A) + H(X,Y|A)+2\epsilon)}\left(\frac{L_3}{2^n}\right)^{\mathcal{C}(\mathcal{V}^{\ast}_{s_{1},s_{2};t})}.
  \end{align}
\end{itemize}
To conclude the achievable region for this case, note that the events $\mathcal{E}_4$ and $\mathcal{E}_6$ yield redundant constraints; thus, the total probability of error tends to zero for a finite size of network,$L_3$ and large $n$ only if the inequalities in \eqref{NC:theoremeq} are satisfied.
\subsection{For the case $I(X;A)-I(Y;A)\leq0$}
\textbf{Error Analysis:}
 The events corresponding to possible encoding and decoding errors in a terminal node $t\in\tau$ are as follows:
\begin{align}
  \mathcal{E}_{1} &= \{\not\exists\vA: (\vx,\vA)\in \strtyp(X,A),\mathcal{B}_{\vA}=\mathcal{B}_{\vx}\}\\
  \mathcal{E}_{2} &= \{(\vX,\vA,\vY)\not\in\ \strtyp(X,A,Y)\}, \\
  \mathcal{E}_{3} &= \{\exists \vtX\neq \vX:
  [\ItX, \IA, \IY]\netz, \mathcal{B}_{\vA}=\mathcal{B}_{\vtX}, (\vtX,\vA,\vY)\in\ \strtyp(X,A,Y)\},\\
  \mathcal{E}_{4} &= \{\exists \vtY\neq \vY:
  [\IX, \IA, \ItY]\netz, (\vX,\vA,\vtY)\in\ \strtyp(X,A,Y)\},\\
  \mathcal{E}_{5} &= \{\exists \vtX\neq \vX,\vtY\neq \vY:
  [\ItX, \IA, \ItY]\netz, (\vtX,\vA,\vtY)\in\ \strtyp(X,A,Y)\},\\
  \mathcal{E}_{6} &= \{\exists \vtX\neq \vX,\vtA\neq \vA,\vtY\neq \vY:
  [\ItX, \ItA, \ItY]\netz, (\vtX,\vtA,\vtY)\in\ \strtyp(X,A,Y)\}.
\end{align}

\begin{itemize}
  \item[1.] $\Pr(\mathcal{E}_1)\rightarrow0$ for $n\rightarrow\infty$ from the covering lemma since each bin $\mathcal{B}_{\vA}$ contains $I(X;A)+\epsilon$ codewords.
  \item[2.]  $\Pr(\mathcal{E}_2|\mathcal{E}_{1}^{C})\rightarrow0$ as $n\rightarrow\infty$ from the same arguments of the case $I(X;A)-I(Y;A)\geq0$.
  \item[3.] To upper-bound $\Pr(\mathcal{E}_{3})$, we have
  \begin{align}
    \Pr(\mathcal{E}_{3}) &= \Pr \left( \exists \vtX \neq \vX:
  [\ItX-\IX,\mathbf{0} , \mathbf{0}]\net =\mathbf{0}, (\vtX,\vA,\vY)\in \strtyp,\mathcal{B}_{\vA}=\mathcal{B}_{\vtX}\right)\nonumber\\
  &= \sum\limits_{(\vx,\va,\vy)}P(\vx,\va,\vy) \nonumber\\
  &\hspace{5mm}\Pr\left(\exists\vtX\neq\vx:
  [\ItX-\Ix, \mathbf{0}, \mathbf{0}]\net=\mathbf{0}, (\vtX,\va,\vy)\in \strtyp,\mathcal{B}_{\va}=\mathcal{B}_{\vtX}\right)\nonumber\\
  &= \sum\limits_{(\vx,\va,\vy)}P(\vx,\va,\vy)\sum_{\vtx\in\mathcal{Q}}\Pr\left([\Itx-\Ix,\mathbf{0} , \mathbf{0}]\net=\mathbf{0}\right),\nonumber\\
  & \text{where $\mathcal{Q}:=\{\vtx:\vtx\neq\vx,(\vtx,\vy,\va)\in \strtyp(X|Y,A),\mathcal{B}_{\va}=\mathcal{B}_{\vtx}$\}} \nonumber\\
%  &= \sum\limits_{(\vx,\va,\vy)} \mspace{-19mu}P(\vx,\va,\vy)\sum_{\vta\in\mathcal{Q}}\sum\limits_{\substack{\vtx\neq\vx:\\ \vtx\in \strtyp(X|Y,A)}}\mspace{-30mu}\Pr\left([\Itx-\Ix, \Ita-\Ia, \mathbf{0}]\net=\mathbf{0}| (\vtx,\vta,\vy)\in \strtyp\right)\\
  &\stackrel{(a)}\leq \sum\limits_{(\vx,\va,\vy)} P(\vx,\va,\vy)|\mathcal{Q}|\left(\frac{L_1}{2^n}\right)^{\mathcal{C}(\mathcal{V}^{\ast}_{s_{1};t})}\nonumber\\
  &\stackrel{(b)}\leq \sum\limits_{(\vx,\va,\vy)} P(\vx,\va,\vy)2^{n(H(X|Y,A)-\Delta+2\epsilon)}\left(\frac{L_1}{2^n}\right)^{\mathcal{C}(\mathcal{V}^{\ast}_{s_{1};t})}\nonumber\\
  &\leq  2^{n(I(X;A)-I(Y;A)+H(X|Y,A)+3\epsilon)}\left(\frac{L_1}{2^{n}}\right)^{\mathcal{C}(\mathcal{V}^{\ast}_{s_{1};t})},
 \end{align}
where:
\begin{itemize}
  \item [(a)] follows from applying Lemma \ref{lemma:bounds}. The notation $L_1$ denotes the maximum path length between $s_1$ and $t$. Note that $\vx\neq\vtx$ implies $\Ix\neq\Itx$ from the same arguments in the analysis of the case $I(X;A)-I(Y;A)\geq0$;
  \item [(b)] follows from deriving an upper bound on $|\mathcal{Q}|$. Namely, we are interested in the amount of source sequences $\vtX$ that are jointly typical with $(\vy,\va)$, moreover the first $n\Delta$ bits of $\Itx$ need to be identical to the bin $\mathcal{B}_{\va}$. The size of this conditional typical set is $2^{n(H(X|Y,A)+2\epsilon)}$, since we know the first $n\Delta$ bits the amount of sequences that fall into this criteria is $2^{n(H(X|Y,A)-\Delta)}$.
  \end{itemize}
\item[4.] To upper-bound $\Pr(\mathcal{E}_{4})$, we have
  \begin{align}
    \Pr(\mathcal{E}_{4}) &= \Pr \left(\exists \vtY\neq \vY:
  [\IX, \IA, \ItY]\netz, (\vX,\vA,\vtY)\in\ \strtyp(X,A,Y)\right)\nonumber\\
  &= \sum\limits_{(\vx,\va,\vy)}P(\vx,\va,\vy)\Pr\left(\exists\vtY\neq\vy:
  [\mathbf{0}, \mathbf{0},\ItY-\Iy]\net=\mathbf{0}, (\vx,\va,\vtY)\in \strtyp\right)\nonumber\\
  &= \sum\limits_{(\vx,\va,\vy)}P(\vx,\va,\vy)\sum_{\vty\in \strtyp(Y|X,A)}\Pr\left([\mathbf{0},\mathbf{0},\Ity-\Iy]\net=\mathbf{0}\right),\nonumber\\
  &\leq \sum\limits_{(\vx,\va,\vy)} P(\vx,\va,\vy)|\strtyp(Y|X,A)|\left(\frac{L_2}{2^n}\right)^{\mathcal{C}(\mathcal{V}^{\ast}_{s_{2};t})}\nonumber\\
  &\leq \sum\limits_{(\vx,\va,\vy)} P(\vx,\va,\vy)2^{n(H(Y|X,A)+2\epsilon)}\left(\frac{L_2}{2^n}\right)^{\mathcal{C}(\mathcal{V}^{\ast}_{s_{2};t})}\nonumber\\
  &\leq  2^{n(H(Y|X,A)+2\epsilon)}\left(\frac{L_2}{2^{n}}\right)^{\mathcal{C}(\mathcal{V}^{\ast}_{s_{2};t})},
 \end{align}
  \item[5.] To upper-bound $\Pr(\mathcal{E}_{5})$, we have
  \begin{align}
  &\Pr(\mathcal{E}_{5}) \nonumber\\
  &= \Pr \left( \exists \vtX\neq \vX,\vtY\neq \vY:
  [\ItX-\IX, \mathbf{0}, \ItY-\IY]\net =\mathbf{0}, (\vtX,\vA,\vtY)\in \strtyp,\mathcal{B}_{\vA}=\mathcal{B}_{\vtX}\right) \nonumber\\
  &= \sum\limits_{(\vx,\va,\vy)} P(\vx,\va,\vy)\cdot\nonumber\\
  &\hspace{5mm} \Pr\left( \exists \vtX\neq\vx,\vtY\neq \vy:
  [\ItX-\Ix, \mathbf{0}, \ItY-\Iy]\net =\mathbf{0}, (\vtX,\va,\vtY)\in \strtyp\right)\nonumber\\
  &\leq \sum_{(\vx,\va,\vy)} P(\vx,\va,\vy)\sum\limits_{\substack{\vtx\neq\vx,\vty\neq\vy:\\(\vtx,\va,\vty)\in \strtyp(X,Y|A)}}\Pr\left([\Itx-\Ix, \mathbf{0}, \Ity-\Iy]\net=\mathbf{0}\right)\nonumber\\
  &\leq \sum\limits_{(\vx,\va,\vy)} P(\vx,\va,\vy)|\strtyp(X,Y|A)|\left(\frac{L_3}{2^n}\right)^{\mathcal{C}(\mathcal{V}^{\ast}_{s_{1},s_{2};t})}\nonumber\\
  &\leq  2^{n(H(X,Y|A)+\epsilon)}\left(\frac{L_3}{2^n}\right)^{\mathcal{C}(\mathcal{V}^{\ast}_{s_{1},s_{2};t})}.
  \end{align}
  \item[6.] To upper-bound $\Pr(\mathcal{E}_{6})$, we have
  \begin{align}
  &\Pr(\mathcal{E}_{6})\nonumber\\
  &=\Pr \left(\exists \vtX\mspace{-10mu}\neq \vX,\vtA\mspace{-10mu}\neq \vA,\vtY\mspace{-10mu}\neq \vY\mspace{-7mu}:
  [\ItX\mspace{-15mu}-\IX, \ItA\mspace{-15mu}-\IA, \ItY\mspace{-15mu}-\IY]\net =\mathbf{0}, (\vtX\mspace{-7mu},\vtA\mspace{-7mu},\vtY\mspace{-4mu})\in \strtyp,\mathcal{B}_{\vtA}\mspace{-4mu}=\mspace{-4mu}\mathcal{B}_{\vtX}\right)\nonumber\\
  &= \sum\limits_{(\vx,\va,\vy)} P(\vx,\va,\vy)\cdot\nonumber \\
  &\Pr\left( \exists \vtX\neq\vx, \vtA\neq\va, \vtY\neq\vy:
  [\ItX\mspace{-15mu}-\Ix, \ItA\mspace{-15mu}-\Ia, \ItY\mspace{-15mu}-\Iy]\net =\mathbf{0}, (\vtX,\vtA,\vtY)\in \strtyp,\mathcal{B}_{\vtA}=\mathcal{B}_{\vtX}\right)\nonumber\\
  &= \sum\limits_{(\vx,\va,\vy)} P(\vx,\va,\vy)\cdot\nonumber\\
  &\sum_{\vtx\in \strtyp(X)}\mspace{-10mu}\Pr\left(\exists \vtA\neq \va,\vtY\neq\vy:
  [\Itx\mspace{-6mu}-\Ix, \Ita\mspace{-6mu}-\Ia, \ItY\mspace{-6mu}-\Iy]\net =\mathbf{0}, (\vtx,\vtA,\vtY)\in \strtyp,\mathcal{B}_{\vtA}=\mathcal{B}_{\vtx}\right)\nonumber\\
  &\stackrel{(a)}= \sum\limits_{(\vx,\va,\vy)} P(\vx,\va,\vy)\sum_{\vtx\in \strtyp(X)}\sum\limits_{\substack{\vty\neq\vy:\\(\vtx,\vta,\vty)\in \strtyp(Y|X,A)}}\Pr\left([\Itx-\Ix, \Ita-\Ia, \Ity-\Iy]\net =\mathbf{0}\right)\nonumber\\
  &\leq \sum\limits_{(\vx,\va,\vy)} P(\vx,\va,\vy)|\strtyp(X)||\strtyp(Y|X,A)|\left(\frac{L_3}{2^n}\right)^{\mathcal{C}(\mathcal{V}^{\ast}_{s_{1},s_{2};t})}\nonumber\\
  &\leq \sum\limits_{(\vx,\va,\vy)} P(\vx,\va,\vy)2^{n(H(X)+\epsilon)}2^{n(H(Y|X,A)+\epsilon)}\left(\frac{L_3}{2^n}\right)^{\mathcal{C}(\mathcal{V}^{\ast}_{s_{1},s_{2};t})}\nonumber\\
  &\leq  2^{n(H(X) + H(Y|X,A)+2\epsilon)}\left(\frac{L_3}{2^n}\right)^{\mathcal{C}(\mathcal{V}^{\ast}_{s_{1},s_{2};t})}\nonumber\\
  &=  2^{n(I(X;A) + H(X,Y|A)+2\epsilon)}\left(\frac{L_3}{2^n}\right)^{\mathcal{C}(\mathcal{V}^{\ast}_{s_{1},s_{2};t})},
  \end{align}
\end{itemize}
where $(a)$ follows from the fact that for a given $\vtx$, there is only one actions codeword denoted by $\vta$ which is jointly typical with $\vtx$ and satisfying $\mathcal{B}_{\vta}=\mathcal{B}_{\vtx}$.

Note that the constraint induced by the event $\mathcal{E}_5$ is redundant, thus, the constraints in \eqref{NC:theoremeq} are sufficient to show that the total probability of error tends to zero as $n$ tends to infinity.
\bibliography{ref}
\bibliographystyle{IEEEtran}
\end{document}